\documentclass[a4paper,12pt]{article}

\usepackage[greek.ancient, main=english]{babel}
\usepackage[utf8]{inputenc}

\usepackage{gensymb}

\usepackage[lighttt]{lmodern}
\usepackage{listings}

\usepackage{lipsum}
\usepackage{courier}
\usepackage{amsfonts}
\usepackage{amsmath}
\usepackage{amssymb}
\usepackage[ruled,lined,boxed,commentsnumbered]{algorithm2e}
\usepackage{csquotes}
\usepackage{amsthm}
\usepackage{tikz}
\usepackage{stackengine}
\usepackage{bm}
\usepackage{mathdots}
\usepackage{comment}
\usepackage{xcolor}
\usepackage{upgreek}
\usepackage{tikz-qtree}
\usepackage{outlines}
\usepackage{enumitem}
\usepackage{graphicx}
\usepackage[export]{adjustbox}
\usepackage{verbatim}
\usepackage{hhline}
\usepackage{booktabs}
\usepackage[cmtip,all]{xy}
\usepackage{latexsym}
\usepackage{mathtools}
\usepackage{varioref}
\usepackage{float}
\usepackage{textcomp}
\usepackage{xparse}
\usepackage{xspace}
\usepackage{url}

\SetKwProg{Fn}{Function}{}{end}\SetKwFunction{FMPPPF}{FnMPPPF}%
\SetAlgoLongEnd

\SetKwInput{KwInput}{Input}
\SetKwInput{KwOutput}{Output}

\newcommand{\longsquiggly}{\xymatrix{{}\ar@{~>}[r]&{}}}

\newcommand{\downleadsto}{%
\mathrel{\reflectbox{\rotatebox[origin=c]{270}{$\leadsto$}}}}

\newcommand*{\myfrown}{\raisebox{.1em}
{\smash{\scalebox{.8}{${\;\smallfrown\;}$}}}}

\DeclareMathAlphabet{\mathbbmsl}{U}{bbm}{m}{sl}

\usetikzlibrary{decorations.text, calc, automata,positioning,arrows, arrows.meta}

\theoremstyle{definition}
\newtheorem{definition}{Definition}[section]
\newtheorem{theorem}{Theorem}[section]
\newtheorem{lemma}[theorem]{Lemma}

\newtheorem{prop}{Proposition}
\theoremstyle{remark}
\newtheorem*{remark}{Remark}

\newtheorem{example}{Example}[section]

\NewDocumentCommand{\gNr}{O{} o}{
  \ensuremath{%
    \upgamma_{\mathbb{N}_r}%
    \IfValueT{#2}{^{#2}}%
    \IfBlankTF{#1}{}{%
      \!\left(#1\right)%
    }%
  }%
}

\NewDocumentCommand{\aNr}{O{} o}{
  \ensuremath{%
    \upalpha_{\mathbb{N}_r}%
    \IfValueT{#2}{^{#2}}%
    \IfBlankTF{#1}{}{%
      \!\left(#1\right)%
    }%
  }%
}

\NewDocumentCommand{\aNrMin}{O{} o}{
  \ensuremath{%
    \upalpha_{\mathbb{N}_{r_{\text{min}}}}%
    \IfValueT{#2}{^{#2}}%
    \IfBlankTF{#1}{}{%
      \!\left(#1\right)%
    }%
  }%
}

\NewDocumentCommand{\gSr}{O{} o}{
  \ensuremath{%
    \upgamma_{\mathbb{S}_r}%
    \IfValueT{#2}{^{#2}}%
    \IfBlankTF{#1}{}{%
      \!\left(#1\right)%
    }%
  }%
}

\NewDocumentCommand{\aSr}{O{} o}{
  \ensuremath{%
    \upalpha_{\mathbb{S}_r}%
    \IfValueT{#2}{^{#2}}%
    \IfBlankTF{#1}{}{%
      \!\left(#1\right)%
    }%
  }%
}

\NewDocumentCommand{\aSrMin}{O{} o}{
  \ensuremath{%
    \upalpha_{\mathbb{S}_{r_{\text{min}}}}%
    \IfValueT{#2}{^{#2}}%
    \IfBlankTF{#1}{}{%
      \!\left(#1\right)%
    }%
  }%
}

\newcommand{\DrMin}{\ensuremath{\mathcal{D}_{r_{\text{min}}}}\xspace}
\newcommand{\DrQmin}{\ensuremath{\mathcal{D}_{r_{\text{qmin}}}}\xspace}

\makeatletter
\def\Ddots{\mathinner{\mkern1mu\raise\p@
\vbox{\kern7\p@\hbox{.}}\mkern2mu
\raise4\p@\hbox{.}\mkern2mu\raise7\p@\hbox{.}\mkern1mu}}
\makeatother

\begin{document}

\title{Recursive Prime Factorizations: Dyck Words as Representations of Numbers}
\author{Ralph [``Tim''] Leroy Childress, Jr.}
\date{\today}
\maketitle

\begin{abstract}
I propose a class of non-positional numeral systems where numbers are represented by Dyck words, with the systems arising from a recursive extension of prime factorization.  After describing two proper subsets of the Dyck language capable of uniquely representing all natural numbers and a superset of the rational numbers respectively, I consider ``Dyck-complete'' languages, in which every member of the Dyck language represents a number.  I conclude by suggesting possible research directions. 
\end{abstract}

\section{Introduction}

My fascination with patterns exhibited in the set of natural numbers $\mathbb{N} = \{0,1,2,3, \ldots\}$ led me to much experimentation and indeed quite a bit of frustration trying to discover and characterize such patterns.
One of the most perplexing problems I encountered was the inherent arbitrariness of positional numeral systems.  Consider the number 520:
\begin{equation}\label{def:520}
520 = 5 \times 10^{2}  \;+\; 2 \times 10^{1} \;+\; 0\times10^{0}.
\end{equation}

The implicit selection of 10 as base, though a convention tracing back to antiquity, reflects an arbitrary choice with consequences for patterns manifested in the representations.
 For instance, a well-known pattern is that if the sum of the digits in a decimal representation of a number is equal to a multiple of 3, then the number itself is divisible by 3; yet that is not the case with base-2 or base-5.  This and many other such patterns may be generalized to apply to numeral systems of any base $\ge 2$, but for most of us the generalization detracts from the immediacy of the realization.  It would be useful if the system of representation did not require any one number to assume undue importance above the others, so that patterns would directly reflect characteristics of those numbers under examination rather than being obscured by the selection of some other number to serve as ``the base.''

There is another drawback inherent in positional numeral systems with regard to their use to identify and characterize patterns among numbers. 
As is evident in Equation~\ref{def:520}, evaluation of a number's positional representation requires three distinct operations, namely exponentiation, multiplication and addition.
   However, the number represented by decimal 520 can be more simply represented by a unique product of prime numbers
\begin{equation} \nonumber
 2^{3} \times 5^{1} \times 13^{1},
\end{equation}
called its \emph{prime factorization},  the evaluation of which does not involve addition.
\begin{remark}
More precisely, the prime factorization of a number is unique up to the order of the factors.  Also, when I say ``the evaluation of which does not involve addition,'' I am referring to addition as a distinct operation in the evaluation;  obviously the multiplication of natural numbers may be viewed as iterated addition. 
\end{remark}

Thus I found myself on a quest to discover systems for representing numbers where the systems, being based upon prime factorization, neither involve the concept of a base nor require addition for evaluation.    As my quest narrowed, I sought such systems with alphabets of the smallest size.

I succeeded in my quest, discovering a class of systems I call ``Natural Recursive Prime Factorizations'' (``Natural RPFs''), where each of these systems can uniquely and exactly represent all members of $\mathbb{N}$ using a language with an alphabet of only two symbols.  I subsequently realized natural RPFs can be extended to yield another class of systems, ``Superrational Recursive Prime Factorizations,'' each of which is capable of uniquely and exactly representing not merely all members  of $\mathbb{N}$ but all members of a superset of the rational numbers.  A remarkable fact about superrational RPF systems is that, unlike decimal and other positional numeral systems, no enlargement of the alphabet is needed beyond that used for representing natural numbers; the same two-symbol alphabet is employed as for  natural RPF, without need for additional symbols such as negative signs, radix points, ellipses, overbars, or underbars.

I must warn the reader at the outset that these systems are impractical for application to the mundane tasks of everyday life, such as balancing checkbooks or enumerating street addresses. But they were never intended for such purposes; rather they were conceived to facilitate the study of patterns among numbers, from a perspective in which the structure of arithmetic is approached through the lens of formal languages and grammars rather than algebraic operations.
Natural RPF systems, for example, invite the analysis of their words using techniques from computer science such as context-free grammars, parsers and finite state machines, providing direct connections between numbers and subsets of the well-studied Dyck language $\mathcal{D}$, including $\mathcal{D}$ itself.   Words produced in these systems moreover do not involve an arbitrarily selected base, eliminate the necessity for addition in their evaluation, and are closely related to the (extended) prime factorizations of almost all the numbers they represent.

\begin{remark}
I say \emph{almost all} because 0 and 1 have no prime factorizations.  Also, note that I regard prime numbers as themselves having prime factorizations, the factorization of a prime number being the number itself, such that the  exponential form of the factorization of the $k$th prime $p_k$ is $p_k^1$.
\end{remark}
Many interesting patterns arise in number sequences defined according to properties shared by their members' representations as Dyck words (see Section~\vref{loc:InvestigateStripes}, for example), and related prime-factor-based representations have recently been used as structured arithmetic texts in studies pertaining to transformer learnability \cite{contucci2025statistical, breccia2025testing}.

\section{The Standard Minimal RPF Natural Interpretation $\text{RPF}_{\mathbb{N}_{r_{\text{min}}}}$}
I begin by presenting a system capable of representing natural numbers by unique finite sequences of left and right  parentheses.      For now I will refer to this system as ``minimal natural RPF,'' abbreviated $\text{RPF}_{\text{min}}$, in order to introduce the concept without first launching into a lengthy digression concerning languages and their interpretations.  In Section~\ref{loc:RPFNMinAndGamma}, I will identify the system more precisely as ``the standard minimal RPF natural interpretation $\text{RPF}_{\mathbb{N}_{r_{\text{min}}}}$.''

\subsection{Informal Treatment by Example}
If challenged to describe minimal natural RPF in one sentence, I might say: ``It is a numeral system in which  0 and 1 are represented by the empty string $\upepsilon$ and $()$ respectively, with every other natural number $n$ being written as a product of powers of consecutive primes from 2 up to and including  the  greatest prime factor of $n$, each exponential term being surrounded by a single pair of parentheses and nonzero exponents being recursively treated in the same fashion as described for $n$, with the final resulting expression being stripped of all symbols except the parentheses, which are then rewritten on one line while preserving their order from left to right."

I myself have difficulty digesting that long-winded sentence; let us instead consider three examples, these being collectively sufficient to suggest how an arbitrary natural number may be represented in minimal natural RPF. To start with, the representations of zero and one are given by explicit definition:
\begin{itemize}
\item  Zero is represented by the empty word $\upepsilon$.
\item One is represented by the word ${()}$.
\end{itemize}
The $\text{RPF}_{\text{min}}$ representation of every other natural number may be obtained by application of a recursive algorithm, as I will illustrate by finding the $\text{RPF}_{\text{min}}$ equivalent of decimal $520$.  But first I must introduce a function that will be used extensively in the algorithm.

\subsubsection{The Minimal Parenthesized Padded Prime Factorization}\label{loc:MPPPF}
We can express 520 as the exponential form of its prime factorization

\begin{equation}\label{exp:PrimeFact520}
p_{1}^{3}p_{3}^{1}p_{6}^{1},
\end{equation}
where $p_{k}$ is the $k$th prime number.  The expression includes powers of $p_{1}$, $p_{3}$ and $p_{6}$, but not of $p_{2}$, $p_{4}$ or $p_{5}$, since these last three do not contribute to the prime factorization of 520.  But let us rewrite Expression~\ref{exp:PrimeFact520} as
\begin{equation}\label{exp:PrimeFact520WithEmbeddedZeroes}
p_{1}^{3}p_{2}^{0}p_{3}^{1}p_{4}^{0}p_{5}^{0} p_{6}^{1},  \nonumber
\end{equation}
so that powers of all consecutive primes $p_{1}, \ldots, p_{m}$ are included, where $p_{m}$ is the greatest prime factor of 520.  Now let us use single pairs of parentheses as grouping symbols around each exponential term, giving
\begin{equation}\label{exp:MinPaddedPrimeFact520}
(p_{1}^{3})(p_{2}^{0})(p_{3}^{1})(p_{4}^{0})(p_{5}^{0})(p_{6}^{1}).
\end{equation}

Expression~\ref{exp:MinPaddedPrimeFact520} is the \emph{minimal parenthesized padded prime factorization} ($\text{MPPPF}$, pronounced ``MIP-fuh'') of 520.  It is \emph{minimal} because only powers of prime numbers up to and including the greatest prime factor of the number being represented are present, it is \emph{parenthesized} for the obvious reason that all exponential terms are enclosed in parentheses, and it is \emph{padded} because it includes exponential terms not appearing in the prime factorization.

While the number 1 has no prime factorization, and while we have already explicitly defined the representation of 1 to be $()$, it will be convenient to speak of MPPPF(1) in our algorithm.  We choose to define MPPPF(1) $ = (p_{1}^{0})$, with the result that the domain of MPPPF is the set of positive integers.

\begin{remark}
Although any prime number raised to the zeroth power equals 1, the selection of MPPPF(1) $ = (p_{1}^{0})$ reflects the fact that  $p_{1} = 2$ is the smallest such prime.
\end{remark}

Observe that there cannot be more than one MPPPF corresponding to a given number, since MPPPF(1) is unique and MPPPFs for all other numbers $n$ in the domain of MPPPF are the result of padding the unique prime factorization of $n$ with the 0th powers of those noncontributing primes less than $n$'s greatest prime factor. 

\subsubsection{Finding the RPF$_{\text{min}}$ Equivalent of Decimal 520} \label{loc:FindingRpfSpellingOf520}

We begin by expressing $520$ as its MPPPF
\begin{equation}
(p_{1}^{3})(p_{2}^{0})(p_{3}^{1})(p_{4}^{0})(p_{5}^{0})(p_{6}^{1}). \nonumber
\end{equation}

Our next step is to replace all nonzero exponents in the expression by their MPPPFs as well.  We repeat this step until there no longer exist opportunities to replace exponents by MPPPFs:

\begin{equation} \label{eqn:ReplacementsOfNonzeroExponentsByMpppfs}
\begin{aligned}
&(p_{1}^{3})(p_{2}^{0})(p_{3}^{1})(p_{4}^{0})(p_{5}^{0})(p_{6}^{1}) =\\ 
&(p_{1}^{     (p_{1}^{0})(p_{2}^{1})        })(p_{2}^{0})(p_{3}^{(p_{1}^{0})})(p_{4}^{0})(p_{5}^{0})(p_{6}^{(p_{1}^{0})}) =\\
&(p_{1}^{(p_{1}^{0})(p_{2}^{(p_{1}^{0})})})(p_{2}^{0})(p_{3}^{(p_{1}^{0})})(p_{4}^{0})(p_{5}^{0})(p_{6}^{(p_{1}^{0})}).\\
\end{aligned}
\end{equation}
Now we proceed to the next and final step, which is to treat the expression as a string and delete all symbols except parentheses from it, writing the parentheses all on one line while preserving their order from left to right to yield the RPF\textsubscript{min} word
\begin{equation} \nonumber \label{exp:RPFMinWordFor520}
(()(()))()(())()()(()).
\end{equation}

We may be certain (()(()))()(())()()(()) is the only minimal natural RPF word corresponding to decimal 520.  This is because the MPPPF of 520 is unique, and each nonzero exponent in Equation Set~\ref{eqn:ReplacementsOfNonzeroExponentsByMpppfs} has exactly one corresponding MPPPF.

\begin{remark}
If we had defined the MPPPF function so that the zeroth powers of primes $p_k$ were expressed as $(1)$ rather than $(p_k^{0})$, our recursive algorithm would still result in 520 being encoded as $(()(()))()(())()()(())$.
\end{remark}

\subsubsection{Finding the Decimal Equivalent of (()(()))()(())()()(())}

Having found the equivalent of decimal 520 in minimal natural RPF, let us go in the reverse direction, finding the  decimal equivalent of the $\text{RPF}_{\text{min}}$ word $(()(()))()(())()()(())$.

We begin by inserting 0 inside each empty matched pair of parentheses, yielding the expression
\begin{equation} \nonumber
((0)((0)))(0)((0))(0)(0)((0)).
\end{equation}

For the next step, we will treat expressions as containing zero or more ``clusters,'' by which I mean substrings beginning and ending with outermost matching parentheses; for example, the clusters from left to right in the expression above are ((0)((0))), (0), ((0)), (0), (0) and ((0)).   For each cluster $w_{k}$, we replace $w_{k}$ by the string
\begin{equation} \nonumber
p_{k}^{contents_{k}},
\end{equation}
where $contents_{k}$ is the string obtained by deleting the outermost parentheses of $w_{k}$.  We do this repeatedly to the successive expressions until all the parentheses are gone:

{
\begin{samepage}
\begin{equation} \nonumber
((0)((0)))(0)((0))(0)(0)((0))
\end{equation}
\begin{equation} \nonumber
\downleadsto
\end{equation}
\begin{equation} \nonumber
p_{1}^{(0)((0))}p_{2}^{0}p_{3}^{(0)}p_{4}^{0}p_{5}^{0}p_{6}^{(0)}
\end{equation}
\begin{equation} \nonumber
\downleadsto
\end{equation}
\begin{equation} \nonumber
p_{1}^{p_{1}^{0}p_{2}^{(0)}}p_{2}^{0}p_{3}^{p_{1}^0}p_{4}^{0}p_{5}^{0}p_{6}^{p_{1}^{0}}
\end{equation}
\begin{equation} \nonumber
\downleadsto
\end{equation}
\begin{equation} \nonumber
p_{1}^{p_{1}^{0}p_{2}^{p_{1}^{0}}}p_{2}^{0}p_{3}^{p_{1}^0}p_{4}^{0}p_{5}^{0}p_{6}^{p_{1}^{0}}.
\end{equation}
\end{samepage}
}

All that remains to be done is to evaluate the expression:
\begin{equation} \nonumber
p_{1}^{p_{1}^{0}p_{2}^{p_{1}^{0}}}p_{2}^{0}p_{3}^{p_{1}^0}p_{4}^{0}p_{5}^{0}p_{6}^{p_{1}^{0}} =  {2}^{3}\cdot{5}\cdot{13} = 520.
\end{equation}

\begin{remark}
In 2012, Devlin and Gnang \cite{devlin2012primes} defined a recursive structure called a \emph{tower factorization}, equal to a product of power towers such that the base and all exponents in the power towers are prime (except for the tower factorization of 1, which is defined to be 1 itself).
Tower factorizations bear a marked resemblance to products of power towers such as the one shown above, and indeed the latter structure may be regarded as an extension of the former one.

\end{remark}

We thus have a system capable of representing every natural number with an alphabet of only two symbols and not involving addition for evaluation.  This fact may not seem particularly significant, given that
the unary system of representing $n$  by $n$ contiguous marks uses an alphabet of only \emph{one} symbol.
But unlike unary, minimal natural RPF directly and succinctly reflects the prime factorizations of the numbers it represents (for those numbers having prime factorizations, which includes all members of $\mathbb{N}$ except 0 and 1).  Indeed, for all natural numbers $n$ greater than  one, the minimal natural RPF word representing $n$ contains not merely the prime factorization of $n$, but also the prime factorizations of all factorizable numbers involved in the prime factorization of $n$, the exponents in the prime factorization themselves being represented by their factorizations in recursive fashion.

\begin{remark}
See Table~\vref{tbl:SpellingsOfNaturals} for the minimal natural RPF representations of the first 20 natural numbers.  An ASCII file containing a longer table may be found among the ancillary files accompanying this paper.
\end{remark}

\subsubsection{A Preview of the Superrational Extension}  \label{loc:preview_of_superrational_extension}

Section~\ref{loc:StdMinRpfRationalInterpretation} explains in detail how minimal natural RPF can be extended to yield a system capable of exactly and uniquely representing all rational numbers, some algebraic irrationals, and even some transcendental numbers, again using the same alphabet of only two symbols.  I limit my discussion here to suggest the key concept upon which the extension is based.

Recall that MPPPF stands for ``Minimal Parenthesized Padded Prime Factorization''; the designation \emph{minimal} refers to the fact that only powers of prime numbers up to and including the greatest prime factor of the number being factorized are present.  Without that restriction, zeroth powers of primes greater than the greatest prime factor could appear.  Suppose we relax the restriction so we can opt to include the zeroth power of $p_{k+1}$, where $p_k$ is the greatest prime factor of the number being factorized.  We could thus express $520$  either minimally as
\enlargethispage{1\baselineskip}
\begin{equation}
(p_{1}^{3})(p_{2}^{0})(p_{3}^{1})(p_{4}^{0})(p_{5}^{0})(p_{6}^{1}) \nonumber
\end{equation}
or nonminimally as
\begin{equation}
(p_{1}^{3})(p_{2}^{0})(p_{3}^{1})(p_{4}^{0})(p_{5}^{0})(p_{6}^{1})(p_{7}^{0}). \nonumber
\end{equation}
Under this scheme, both $(()(()))()(())()()(())$ and $(()(()))()(())()()(())()$ evaluate to the natural number $520$, as empty parenthesis pairs correspond to $1$s.  Moreover, we know that a Dyck word ending with ${()}$---other than the word $()$ itself---does not require a terminal ${()}$ for its evaluation as a natural number.  This allows us to regard the final ${()}$ as a negative sign:
\begin{equation}
(()(()))()(())()()(())  \equiv 520, \nonumber
\end{equation}
while
\begin{equation}
(()(()))()(())()()(())()  \equiv -520. \nonumber
\end{equation}

Thus an optional empty pair of matched parentheses may be regarded as encoding sign.  But with recursive prime factorizations, sign propagates to exponents; as a result, negative powers are permitted.  For example, we shall see on page \pageref{ex:SpellingOfGelfondSchneider} that the transcendental number $2^{\sqrt{2}} = 2^{2^{2^{-1}}} $ may be represented as $(((()())))$.

\begin{remark}
Ancillary Jupyter notebooks are provided for both systems.
\end{remark}

\subsection{Interpretations and Representations}\label{loc:LanguagesAndInterpretations}
Recall that I introduced Equation~\ref{def:520} by writing ``Consider the number 520.'' My wording was intended as a device to illustrate an important point in the present section.  So conflated in our minds are numbers with their representations, and with their decimal representations in particular, that I suspect few readers were bothered by the phrase ``the number 520'' as being meaningless, or at best an incomplete abbreviation of  ``the number represented by 520 in the decimal numeral system.''  That is to say, most of us seldom stop to distinguish between \emph{numbers} and \emph{number words}.  But there is in fact a distinction, and to ignore it can yield untoward consequences.  For example, the set of natural numbers contains a unique multiplicative identity element $1$ such that $1\cdot n = n = n \cdot 1$  for all $n \in \mathbb{N}$.
But all of the members of the following set are in the decimal system, and all evaluate to $1$: $\{1, 01, 001, 0001, \ldots\}$.  There are then infinitely many ``decimal numbers'' (decimal number words) that can be considered the identity element for multiplication.  Thus we might be tempted to conclude that the set of natural numbers contains infinitely many multiplicative identity elements, despite the existence of simple proofs to the contrary.

It is especially important that we maintain the distinction in this paper, which is intimately concerned with numbers and different ways of representing them.  A sequence of symbols is one thing; what that sequence means is quite another.  For example, 11 can be understood to mean $11_{10}$ in decimal, $3_{10}$ in binary, or $17_{10}$ in hexadecimal.

I find ``meaning'' a difficult concept to state with precision, so instead I offer a definition of the word ``interpretation.''

\begin{definition} \label{def:Interpretation}
Let $L$ be a formal language, and let $S$ be a set.  If there exists some surjective function $f: L \rightarrow S$, 
then the triplet $(L,S,f)$ is \emph{an interpretation of L as S}, specifically \emph{the interpretation of L as S according to f}, and we may say any of the following:
\begin{itemize}
\item $L$ interpreted according to $f$ is $S$ (equivalently: $L$ is the \emph{underlying language} in the interpretation $(L, S, f)$).
\item $S$ is $L$ interpreted according to $f$ (equivalently: $S$ is the \emph{target set} in the interpretation $(L,S,f)$).
\item $f$ interprets $L$ as $S$ (equivalently: $f$ is the \emph{evaluation function} in the interpretation $(L,S,f)$).
\end{itemize}
\end{definition}

As an example, let $B$ denote the set of nonempty strings over the alphabet $\{0,1\}$, and let ${f : B \rightarrow \mathbb{N}}$, where $f({b})$ is the nonnegative integer corresponding to $b$ such that the latter is regarded as a word in unsigned binary.  Then ${(B,\mathbb{N}, f)}$ is an interpretation of $B$ as the set of natural numbers, specifically the interpretation of $B$ as the set of natural numbers according to  $f$.  Now consider ${g : B \rightarrow \mathbb{Z}}$, where $g({b})$ is the integer corresponding to $b$ such that the latter is regarded as a word in 2s-complement binary, with the qualification that nonnegative integers always correspond to words containing the prefix 0.  Then ${(B,\mathbb{Z}, g)}$ is an interpretation of $B$ as the set of integers, specifically the interpretation of $B$ as the set of integers according to $g$.  Thus we see that the same language may underlie multiple interpretations.

The following definition allows us to speak of interpretations in terms of set members as well.
\begin{definition} \label{def:InterpretationAsAppliedToMembers}
Let $(L,S,f)$ be an interpretation.  For any $l \in L$, let $s \in S$ such that $s = f(l)$. Then we may say any of the following:
\begin{itemize}
\item $s$ is $l$ interpreted according to $f$.
\item $l$ interpreted according to $f$ is $s$.
\item $f$ interprets $l$ as $s$.
\end{itemize}

\end{definition}

Thus the encoding of unsigned binary interprets 11111111 as decimal 255, whereas the same number word interpreted according to the encoding of signed binary is decimal -1.

Sometimes the terminology of interpretations becomes awkward, resulting in a surfeit of passive participles  (``$l$ interpreted as...''; ``$l$ interpreted according to...'').  We may remedy this to some extent by making use of the following definitions.

\begin{definition} \label{def:LanguageRepresent}
If $(L,S,f)$ is an interpretation, then we say \emph{L represents $S$ in $(L,S,f)$}, or, equivalently, \emph{$L$ is a representation of $S$ in $(L,S,f)$}.  In cases where we do not wish to mention a specific interpretation, we may simply say \emph{$L$ represents $S$}, or, equivalently, \emph{$L$ is a representation of $S$}.
\end{definition}

\begin{definition} \label{def:Represent}
Let $(L,S,f)$ be an interpretation, let $l \in L$, and let $s$ be a member of $S$ satisfying the equation $s = f(l)$.  Then we say \emph{l represents s in (L,S,f)}.  If we do not wish to mention a particular interpretation (as for example when the interpretation would be clear from the context), we may simply say \emph{l represents s}, implying some interpretation exists such that $l$ represents $s$ in that interpretation.
\end{definition}

Definition~\ref{def:Interpretation} only requires that the function $f$ be surjective.  I introduce special terminology for the case where $f$ is injective as well.

\begin{definition} \label{def:MinimalInterpretation}
Let $(L,S,f)$ be an interpretation such that $f$ is a bijection.  Then we may say any of the following:
\begin{itemize}
\item $(L,S,f)$ is \emph{minimal}.
\item \emph{L is minimal with respect to (L,S,f)}.
\item \emph{The representation of $S$ in $(L,S,f)$ is minimal}.
\item \emph{L is minimal}.
\end{itemize}
If $L$ is also the language underlying another interpretation $(L,T,g)$ such that $g$ is not a bijection, we may simply say \emph{L is quasiminimal}, implying two intepretations exist such that  $L$ is minimal with respect to one but not with respect to the other.
\end{definition}

If an interpretation $(L,S,f)$ is minimal, there is exactly one member of $L$ representing any given  $s \in S$;  otherwise we cannot exclude the possibility that $s$ may have multiple representations.  Regardless of whether the interpretation is minimal, the surjectivity of $f$ ensures that every member of $S$ has at least one representation in $L$.

The following definition clarifies what I understand by the word \emph{system} when I speak of RPF systems.
\begin{definition} \label{def:NumeralSystem}
Let $I$ be an interpretation such that the target set in $I$ is numerical.  Then we say $I$ is a \emph{numeral system}. We may simply say ``$I$ is a system,'' if doing so does not incur ambiguity.
\end{definition}

\subsubsection{Standard and Other Prime-permuted RPF Interpretations}

RPF systems are less arbitrary than positional numeral systems, not requiring the selection of a special number to serve as the base.  Yet some particular permutation of the sequence of prime numbers must be chosen for the interpretation of RPF words as numbers and  the representation of numbers by RPF words.
Returning momentarily to the concept of minimal parenthesized padded prime factorizations (Section~\ref{loc:MPPPF}), MPPPFs involve powers of primes appearing in the same order as those primes occur in the sequence
$(2,3,5, \ldots, p_{k})$, where $p_{k}$ is the greatest prime factor of the number of which the MPPPF is taken.
But the descending sequence $(p_{k}, \ldots, 5,3,2)$ could be used to yield a ``reverse MPPPF,''  so that RPF words produced using reverse MPPPFs would be mirror images of their equivalents as produced using MPPPFs---I mean ``mirror images'' literally, in the sense that if we wrote down a reverse-MPPPF-derived word on a piece of paper and viewed it in a mirror, we would see in the reflection an image identical to its unreflected MPPPF-derived counterpart.  We can therefore use either left-ascending or right-ascending sequences of primes upon which to base representations.  But we need not stop there; we could use the shortest-length sequence $S_{\text{swap}}$  in $(p_{2}, p_{1}, \ldots, p_{2j}, p_{2j-1}, \ldots)$ such that $p_{2}$ was the first term in $S_{\text{swap}}$ and all the primes from the least to the greatest prime factor of the argument of MPPPF were in $S_{\text{swap}}$.  In fact, any permutation of the sequence of prime numbers would suffice to determine the ordering of the prime powers in the definition of a proposed MPPPF.  Nevertheless, in order to avoid a profusion of symbols designating the choice of the underlying prime permutation, and to have common ground for discussing RPF systems,
it would be well to consider one sequence as the ``standard,'' with other permutations only being talked about when their existence was relevant to the discussion.

I select the \emph{identity permutation} $\bm{P} = (2,3,5,7,11 \ldots)$ of prime numbers, where the terms appear in the same order as they occur in the sequence of natural numbers, to be the standard permutation.  Indeed, we can regard $\bm{P}$ as not being a permutation at all, but the original sequence from which other prime sequences are derived by scrambling the terms in $\bm{P}$.  Because values of successive terms in  $\bm{P}$  increase as the terms are written in customary order from left to right, a standard RPF system can also be called a \emph{right-ascending} or \emph{rightwise} system.  This is why I will often refer to a standard RPF system using the subscript $r$, as in $\text{RPF}_{\mathbb{N}_{r_{\text{min}}}}$ and $\text{RPF}_{\mathbb{S}_{r_{\text{min}}}}$.

\begin{definition} \label{def:StandardPermutation}
The \emph{standard permutation}, also called the \emph{rightwise} or \emph{right-ascending permutation}, is the sequence $\bm{P}$ of prime numbers $(2,3,5,7,11, \ldots)$.
\end{definition}

\begin{remark}
Here we have an illustration of how the intimate relationship between prime factorizations and recursive prime factorizations results in parallels between the two.  The fundamental theorem of arithmetic states that a number's prime factorization is unique \emph{up to the order of the factors}, so that the prime factorization of 520 could be written variously as $2^3 \times 5^1 \times 13^1$, $13^1 \times 5^1 \times 2^3$,  etc. But in practice we usually write prime factorizations with the primes appearing in ascending order.
\end{remark}

In this paper I will sometimes refer to a rightwise interpretation without displaying the permutation subscript $r$, with the understanding that the underlying permutation is a rightwise one.  This practice is conceptually similar to that in which the decimal number word $123$ is implicitly understood to designate $123_{10}$. However, of course, the notation I introduce is only applicable to the scope of this paper.

On occasion I will designate an RPF interpretation based upon a particular but arbitrary permutation of the prime number sequence (with the identity permutation being one of the possibilities); in such cases I will use lowercase Greek letters in the notation, as with $\text{RPF}_{\mathbb{N}_{\upsigma_{\text{min}}}}$.

In all of this, we must take care to remember that when we speak of ``permutations,'' we are not referring to the RPF systems themselves but rather to the sequences of prime numbers underlying them.  I will therefore not refer to an RPF interpretation or its language as being a permutation, but rather as being \emph{prime-permuted}, or, if I am referring to a specific prime permutation $\upsigma$, as being \emph{$\upsigma$-permuted}.

\subsection{$\boldsymbol{\upgamma_{\mathbb{N}_{r}}}$ and $\text{RPF}_{\mathbb{N}_{r_{\text{min}}}}$} \label{loc:RPFNMinAndGamma}
\begin{remark}
This section is confined to a discussion of mathematical objects relevant to those RPF systems arising from the standard permutation.  See Section~\ref{loc:GeneralizationOfStdMinToAllPermutations} for generalizations of the same objects for all prime permutations.
\end{remark}

We can think of an RPF representation of a number as being a spelling of that number, with a spelling function mapping from the set of numbers onto the set of words in the language underlying the corresponding RPF interpretation.  If the interpretation is minimal, we can then use the resulting possible spellings to define the language as the image of the spelling function.

We will employ a modification of this approach to define the standard minimal RPF natural interpretation $\text{RPF}_{\mathbb{N}_{r_{\text{min}}}}$, as follows.  We will define a function $\upgamma'_{\mathbb{N}_{r}}$ (from \begin{greektext}γ\end{greektext} in \begin{greektext}ὀρθογραφία\end{greektext}, \emph{orthographia}, Greek for ``word'') mapping a natural number $n$ to a unique string of parentheses, and then we will define the language underlying $\text{RPF}_{\mathbb{N}_{r_{\text{min}}}}$ as the set of all possible strings produced by the function.  Since  $\upgamma'_{\mathbb{N}_{r}}$ takes a natural number input and outputs an $\text{RPF}_{\mathbb{N}_{r_{\text{min}}}}$ spelling of the number, we could choose to regard $\upgamma'_{\mathbb{N}_{r}}$ as our spelling function; however, because we wish the spelling function to have an inverse, we will instead define a spelling function $\upgamma_{\mathbb{N}_{r}}$ identical to $\upgamma'_{\mathbb{N}_{r}}$ except with its codomain restricted to $\text{RPF}_{\mathbb{N}_{r_{\text{min}}}}$.   Thus $\upgamma_{\mathbb{N}_{r}}$ will be a bijection, enabling us to speak of its inverse.

First I introduce two notations for convenience in concatenating strings; these will find extensive use throughout the rest of this paper.

\begin{definition} \label{def:Frown}
The symbol $\myfrown$ is the string concatenation operator; $a \myfrown b$ denotes the concatenation of strings $a$ and $b$.  The concatenation of $a$ and $b$ may also be written in customary fashion as $ab$, provided that doing so incurs no ambiguity.
\end{definition}

\begin{definition}\label{def:BigOPlus}
Let $j, k \in \mathbb{N}_{+}$.  Then we shall understand $\underset{i = j}{\stackrel{k}{\bigoplus}} {s_{i}}$ to mean the string concatenation $s_{j} \ldots{s_{k}}$ if $j \le k$; otherwise the concatenation is the null string $\upepsilon$.
\end{definition}

Without further ado, let us define the nonsurjective precursor to our spelling function.

\begin{definition}\label{def:NonsurjectiveTranscriptionFunction}  Let $\Sigma^{*}$ be the Kleene closure of the set $\{(,)\}$.  Then the \emph{standard nonsurjective RPF natural transcription function}, denoted by   $\upgamma'_{\mathbb{N}_{r}}$, is given by  $\upgamma'_{\mathbb{N}_{r}} : \mathbb{N} \rightarrow \Sigma^{*}$, where
   \begin{itemize}
      \item For $ n = 0$, $\upgamma'_{\mathbb{N}_{r}}(n)$ is the empty string $\upepsilon$.
      \item For $n = 1$, $\upgamma'_{\mathbb{N}_{r}}(n)$ is the string ${()}$.
      \item For $n > 1$, let $p_{m}$ be the greatest prime factor of $n$, and let $a = (a_{1}, \ldots, a_{m})$ be the integer sequence satisfying the equation
      \begin{equation} \nonumber
      n = \prod_{i = 1}^{m} p_{i}^{a_{i}}.
      \end{equation}
      Then
         \begin{equation}
            \upgamma'_{\mathbb{N}_{r}}(n)   ={ } \underset{i = 1}{\stackrel{m}{\bigoplus}} {(\;{'('}\myfrown       \upgamma'_{\mathbb{N}_{r}} (a_{i}})\myfrown{')'\;)}. \nonumber
        \end{equation}
   \end{itemize}
\end{definition}
\begin{equation} \nonumber  
{}
\end{equation}

Now we can define our bijective spelling function by specifying its graph.
\begin{definition} \label{def:NaturalSpellingFunction}
\emph{The standard RPF natural spelling function},
denoted by $\upgamma_{\mathbb{N}_{r}}$, is given by
\begin{equation}\nonumber
\upgamma_{\mathbb{N}_{r}} = \{(n,w) \in \mathbb{N} \times
\upgamma'_{\mathbb{N}_{r}}(\mathbb{N})
 \mid w = \upgamma'_{\mathbb{N}_{r}}(n)\},
\end{equation}
where $\upgamma'_{\mathbb{N}_{r}}(\mathbb{N})$ is the image of $\upgamma'_{\mathbb{N}_{r}}$.
\end{definition}

\begin{remark}
Definition~\ref{def:NaturalSpellingFunction} gives us our spelling function, but relies upon Definition~\ref{def:NonsurjectiveTranscriptionFunction} for computation of its values.
\end{remark}

At times it will be convenient to avoid function application notation to express RPF spellings.
\begin{definition} \label{def:SpellFn_no_parens}
Let $(L,S,f)$ be an interpretation such that the representation of $S$ is minimal, let $g$ be the inverse of $f$, and let $y \in S$.  Then the expression $g_{S_y}$ denotes $g(y)$.
\end{definition}

\begin{remark}
A reference implementation of Definition~\ref{def:NaturalSpellingFunction}  may be found among the ancillary files accompanying this paper.
\end{remark}

We can therefore write \gNr[k] as $\upgamma_{\mathbb{N}_{r_k}}$.

\begin{example}
Let us find the spelling corresponding to decimal 2646.

First we note that
\begin{equation} \nonumber
\begin{aligned}
2646 &= 2^{1} \cdot 3^{3} \cdot 5^{0} \cdot 7^{2 }  \\
        &= p_{1}^{1} \cdot p_{2}^{3} \cdot p_{3}^{0} \cdot p_{4}^{2}.
\end{aligned}
\end{equation}

Thus the standard RPF natural spelling of $2646_{10}$ is

\begin{equation} \label{eqn:Spell2646}
\begin{aligned}
\upgamma_{\mathbb{N}_{r_{2646}}}  & = (\upgamma_{\mathbb{N}_{r_{1}}})
(\upgamma_{\mathbb{N}_{r_{3}}})
(\upgamma_{\mathbb{N}_{r_{0}}})(\upgamma_{\mathbb{N}_{r_{2}}})\\
 & = (())((\upgamma_{\mathbb{N}_{r_{0}}})(\upgamma_{\mathbb{N}_{r_{1}}}))()((\upgamma_{\mathbb{N}_{r_{1}}})) \\
 & = (())(()(()))()((())).
\end{aligned}
\end{equation}
\end{example}

Before we define the standard minimal RPF natural interpretation, let us give a name to its underlying language.

\begin{definition} \label{def:LDMin}
The \emph{standard minimal RPF language}, denoted by $\mathcal{D}_{r_{\text{min}}}$, is the codomain of $\upgamma_{\mathbb{N}_{r}}$.
\end{definition}
\begin{remark}
The presence of the symbol $\mathcal{D}$ in the notation is not intended to signify that $\mathcal{D}_{r_{\text{min}}}$ is a Dyck language; it is rather intended to suggest a relationship between the standard minimal RPF language and the Dyck language. This relationship is the subject of Section~\ref{loc:DyckLangAndDMin}.  Also, notice that I neither included $\mathbb{N}$ in $\mathcal{D}_{r_{\text{min}}}$ nor used the word \emph{natural} in the appellation \emph{standard minimal RPF language}.  That is because the language has a nonnumerical alternative definition, as we will see in Section~\ref{loc:StdMinRPFLangReconsidered}.
\end{remark}

At long last we arrive at the definition of $\text{RPF}_{\mathbb{N}_{r_{\text{min}}}}$.

\begin{definition}\label{def:StdMinimalRPFNaturalInterpretation}
   The \emph{standard minimal RPF natural interpretation}, denoted by  $\text{RPF}_{\mathbb{N}_{r_{\text{min}}}}$, is the interpretation $(\mathcal{D}_{r_{\text{min}}}, \mathbb{N}, \upgamma^{-1}_{\mathbb{N}_{r}})$.
\end{definition}

\begin{table}[htb]
\centering
\begin{tabular}{@{}llll@{}}
\toprule
Decimal & $\text{RPF}_{\mathbb{N}_{r_{\text{min}}}}${\,\,\,\,\,\,\,\,\,\,\,\,\,\,\,\,\,\,\,\,\,\,\,\,\,\,\,\,}        & Decimal & $\text{RPF}_{\mathbb{N}_{r_{\text{min}}}}$                \\ \midrule
0       & $\upepsilon$          & 10      & (())()(())         \\
1       & ()         & 11      & ()()()()(())       \\
2       & (())       & 12      & ((()))(())         \\
3       & ()(())     & 13      & ()()()()()(())     \\
4       & ((()))     & 14      & (())()()(())       \\
5       & ()()(())   & 15      & ()(())(())         \\
6       & (())(())   & 16      & (((())))           \\
7       & ()()()(()) & 17      & ()()()()()()(())   \\
8       & (()(()))   & 18      & (())((()))         \\
9       & ()((()))   & 19      & ()()()()()()()(()) \\ \bottomrule
\end{tabular}
\caption{Standard RPF spellings of the first twenty natural numbers.}
\label{tbl:SpellingsOfNaturals}
\end{table}

\subsection{The Language $\mathcal{D}$ and the Dyck Natural Numbers} \label{loc:DyckLangAndDMin}
 
Table~\vref{tbl:SpellingsOfNaturals} shows the standard minimal RPF spellings of the first twenty natural numbers, suggesting a resemblance between words in $\mathcal{D}_{r_{\text{min}}}$ and those in the Dyck language $\mathcal{D}$; indeed, $\mathcal{D}_{r_{\text{min}}}$ is a proper subset of $\mathcal{D}$.  In light of this fact, and in light of the relevance of the Dyck language to all RPF systems, I provide here a brief description of $\mathcal{D}$.
\begin{remark}
We should avoid drawing too many conclusions from merely looking at the RPF spellings in Table~\ref{tbl:SpellingsOfNaturals}.  For instance, we might conjecture that the RPF spelling of every natural number greater than 0 is longer than its decimal counterpart; however, we may disprove the conjecture by citing the counterexample
\begin{equation} \nonumber
                    \upgamma_{\mathbb{N}_r}(443426488243037769948249630619149892803)  =   {()(()(((()))))}.
\end{equation}
\end{remark}

Recall Section~\ref{loc:FindingRpfSpellingOf520}, where we found an RPF representation of the natural number represented by decimal $520$.  In doing so, we wrote 520 as the expression
\begin{equation} \nonumber
(p_{1}^{(p_{1}^{0})(p_{2}^{(p_{1}^{0})})})(p_{2}^{0})(p_{3}^{(p_{1}^{0})})(p_{4}^{0})(p_{5}^{0})(p_{6}^{(p_{1}^{0})})
\end{equation}
and then deleted everything except the parentheses to yield
\begin{equation} \nonumber
(()(()))()(())()()(()).
\end{equation}
The result was a word in the Dyck language $\mathcal{D}$, the set of all strings consisting of zero or more well-balanced parenthesis pairs.  Informally, we can think of the property of being well-balanced as what distinguishes a syntactically correct use of grouping parentheses in an algebraic expression from a syntactically incorrect one.  For example, the expression 
$(p_{1}^{(p_{1}})$ is nonsensical, the number of closing parentheses not equaling the number of opening parentheses; thus the string $(()$ is not well-balanced and is  not a word in the Dyck language.  The expression 
${p_{4}^{0})(p_{5}^{0}}$ is also nonsensical, even though the number of opening and closing parentheses is equal, because the closing parenthesis is not preceded by a matching opening parenthesis.  Thus the string $)($, not being well-balanced, is not a word in $\mathcal{D}$.
\begin{remark}
The symbols in the language do not have to be parentheses; various authors use parentheses, square brackets, 1s and 0s, etc.  Any binary symbol set will suffice.
\end{remark}

\begin{definition} \label{def:DyckLanguage}
Let $\Sigma^{*}$ be the Kleene closure of the alphabet $\{(,)\}$.  
The \emph{Dyck language}, denoted by $\mathcal{D}$, is the set of all $w \in  \Sigma^{*}$ such that the number of right parentheses in any prefix $w'$ of $w$ does not exceed the number of left parentheses in $w'$, and the number of left parentheses in $w$ is equal to the number of right parentheses in $w$.
\end{definition}

The Dyck language is context-free and is generated by the following grammar, where $\upepsilon$ is the trivial Dyck word containing no symbols:
\begin{equation} \label{eqn:DyckCFG} \nonumber
S \;\rightarrow\;  {(}S{)}S\; \mid \; \upepsilon.
\end{equation}

Now that I have provided a bare-bones introduction to $\mathcal{D}$, I return to my discussion concerning the relationship between that language and $\mathcal{D}_{r_{\text{min}}}$.

The function $\upgamma_{\mathbb{N}_{r}}$ is recursive.  In order that I might discuss sequences of recursive function evaluations, I now develop special notation and vocabulary.  Recall Equation Set~\ref{eqn:Spell2646}, which showed the steps in the evaluation of $\upgamma_{\mathbb{N}_{r}}(2646)$:
\begin{equation} \nonumber
\begin{aligned}
\upgamma_{\mathbb{N}_{r_{2646}}}  & = (\upgamma_{\mathbb{N}_{r_{1}}})
(\upgamma_{\mathbb{N}_{r_{3}}})
(\upgamma_{\mathbb{N}_{r_{0}}})(\upgamma_{\mathbb{N}_{r_{2}}})\\
 & = (())((\upgamma_{\mathbb{N}_{r_{0}}})(\upgamma_{\mathbb{N}_{r_{1}}}))()((\upgamma_{\mathbb{N}_{r_{1}}})) \\
 & = (())(()(()))()((())).
\end{aligned}
\end{equation}

We can use these equations to construct a tree where every node represents an invocation of the spelling function to evaluate a number, with children of the node appearing in the same order as their invocations occur in the node's evaluation:
\medskip

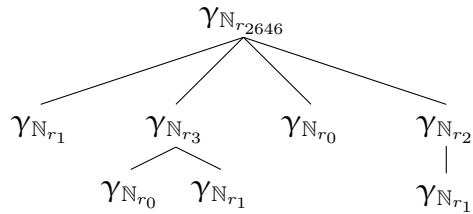
\begin{center}
{
\Tree [.$\upgamma_{\mathbb{N}_{2646}}$ [.$\upgamma_{\mathbb{N}_{1}}$ ] [.$\upgamma_{\mathbb{N}_{3}}$ $\upgamma_{\mathbb{N}_{0}}$ $\upgamma_{\mathbb{N}_{1}}$ ] $\upgamma_{\mathbb{N}_{0}}$ [.$\upgamma_{\mathbb{N}_{2}}$ $\upgamma_{\mathbb{N}_{1}}$ ] ] 
} \label{loc:Recursion_tree_for_gNr_2646}
\end{center}
\medskip

And so we see that the invocation of  $\upgamma_{\mathbb{N}_{r}}$ to spell 2646 leads to $\upgamma_{\mathbb{N}_{r}}$ being invoked to spell 1, followed by $\upgamma_{\mathbb{N}_{r}}$ being invoked to spell 3, followed by $\upgamma_{\mathbb{N}_{r}}$ being invoked to spell 0, followed by $\upgamma_{\mathbb{N}_{r}}$ being invoked to spell 2, with the spellings of 3 and 2 leading recursively to further invocations.

\begin{definition} \label{def:RecursionTree}
Let $S$ be a set, let $f$ be a unary function such that the domain of $f$ is $S$, and let $s \in S$.  Then the \emph{recursion tree for the evaluation of} $f(s)$, also called the  \emph{recursion tree of} $f(s)$, is a tree $T$ in which every node represents an invocation of $f$ to evaluate a member of $S$ such that the invocation arises from the evaluation of $f(s)$, with the root of $T$ representing the invocation $f(s)$, the children of a node having the same order as they occur in the invocation, and all invocations of $f$ arising from the evaluation of $f(s)$ appearing in $T$.
\end{definition}

\begin{definition} \label{def:EntailsAndDirectlyEntails}
Let $f$ be a unary function such that the domain of $f$ is some set $S$, and let $T$ be the recursion tree for $f(s)$, where $s \in S$. 
Suppose there exist distinct nodes $A$ and $B$ in $T$ such that $A$ represents the invocation of $f$ to evaluate some $a \in S$ and $B$ represents the invocation of $f$ to evaluate some $b \in S$.
 If $B$ is a descendant of $A$, then we say $f(a)$ \emph{entails} $f(b)$.  If $B$ is a child of $A$, we may also say $f(a)$ \emph{directly entails} $f(b)$.
\end{definition}
\begin{example}
Referring to the recursion tree for the natural RPF spelling of 2646, we see that the spelling of 2646 entails the spelling of 0, since the spelling of 0 is a descendant of the spelling of 2646.  Moreover, the spelling of 2646 \emph{directly} entails the spelling of 0, since one of the children of the spelling of 2646 is the spelling of 0.  However, the spelling of 3 does not entail the spelling of 2, or vice versa; although both are nodes in the tree, neither is a descendant of the other.
\end{example}

\begin{definition} \label{def:RecursionChain}
Let $f$ be a unary function such that the domain of $f$ is some set $S$, and let $T$ be a recursion tree such that $(f(s_{1}), \ldots, f(s_{k}))$ is the sequence of nodes in a path from the root of $T$ to one of its leaves.
Then we call $(f(s_{1}), \ldots, f(s_{k}))$ a \emph{recursion chain} in the evaluation of $f(s_{1})$.
\end{definition}

\begin{example}
Referring to the recursion tree shown earlier for the natural RPF spelling of 2646, we see that $\upgamma_{\mathbb{N}_{r_{2646}}}$ directly entails $\upgamma_{\mathbb{N}_{r_{3}}}$, which in turn directly entails $\upgamma_{\mathbb{N}_{r_{0}}}$.  Furthermore, the tree is rooted at $\upgamma_{\mathbb{N}_{r_{2646}}}$, and $\upgamma_{\mathbb{N}_{r_{0}}}$ is a leaf.  Therefore $(\upgamma_{\mathbb{N}_{r_{2646}}},   \upgamma_{\mathbb{N}_{r_{3}}}, \upgamma_{\mathbb{N}_{r_{0}}})$ is a recursion chain in the evaluation of $\upgamma_{\mathbb{N}_{r_{2646}}}$.  Since $\upgamma_{\mathbb{N}_{r_{2646}}} = (())(()(()))()((()))$,  $\upgamma_{\mathbb{N}_{r_{3}}} = ()(())$, and $\upgamma_{\mathbb{N}_{r_{0}}} = \upepsilon$, we can also write the recursion chain as $({'}(())(()(()))()((())){'}, {'}()(()){'}, \upepsilon)$.
\end{example}

Now we have sufficient notation and vocabulary to prove that the standard minimal RPF language is a proper subset of the Dyck language.

\begin{theorem} \label{thm:StdMinNaturalRPFIsProperSubsetOfD}
The language $\mathcal{D}_{r_{\text{min}}}$ is a proper subset of the Dyck language $\mathcal{D}$.
\end{theorem}
\begin{proof}
In order for  $\mathcal{D}_{r_{\text{min}}}$ to be a proper subset of $\mathcal{D}$, every member of the former must also be a member of the latter, and at least one member of the latter must not be a member of the former.  Let us address these criteria separately:
\begin{itemize}
\item 
Definition~\ref{def:LDMin} identifies  $\mathcal{D}_{r_{\text{min}}}$ as the codomain of $\upgamma_{\mathbb{N}_{r}}$, which in turn is the image of the standard nonsurjective RPF natural transcription function  $\upgamma'_{\mathbb{N}_{r}}$ according to Definition~\ref{def:NonsurjectiveTranscriptionFunction}.  That latter definition contains three cases for the spelling of a natural number $n$. The first case gives the spelling of 0 as $\upepsilon$, while the second case gives the spelling of 1 as $()$; both of these are members of $\mathcal{D}$.
For every other natural number $n$, the recursion tree of the spelling of $n$ must have as its leaves members of the set ${\{\upgamma_{\mathbb{N}_{r_{0}}}, \upgamma_{\mathbb{N}_{r_{1}}}\}}$, because only the spellings of 0 and 1 do not entail further spellings.  The parent of a leaf in the recursion tree of the spelling of $n$ is a spelling where the leaf and its siblings are each surrounded by single pairs of matched parentheses, the parenthesized expressions then being concatenated in the same order as the siblings appear in the tree.  Enclosing a Dyck word in matched single parentheses results in a Dyck word, and the concatenation of Dyck words also results in a Dyck word; thus the parents of leaves are spellings of Dyck words.  The same process of surrounding children by parentheses and concatenating the resulting expressions to spell the parent applies at every level of the tree, so the spelling of $n$ is a Dyck word.  Thus the spellings of all members of $\mathbb{N}$ are Dyck words, allowing us to conclude that $\mathcal{D}_{r_{\text{min}}} \subseteq \mathcal{D}$.
\item The string $s = {'(())()'}$ will suffice as a counterexample to demonstrate that at least one member of $\mathcal{D}$ is not a member of  $\mathcal{D}_{r_{\text{min}}}$.  The string satisfies the definition of a Dyck word: it contains no symbols other than left and right parentheses, with the number of left parentheses equal to the number of right parentheses and with every prefix of $s$ containing at least as many left parentheses as right parentheses.  However, $s$ is not a member of $\mathcal{D}_{r_{\text{min}}}$, as it is not in the codomain of $\upgamma_{\mathbb{N}_{r}}$.  To see why the spelling function is incapable of producing $s$, we observe that $s$ contains the suffix ${'()'}$ and then try to find a number $n$ such that its spelling contains that suffix.  Let us treat the possibilities individually as follows.  Certainly the spelling of natural number 0 does not contain $'()'$ as a suffix, since $\upgamma_{\mathbb{N}_{r}}(0) = \upepsilon$. While it is true that  $\upgamma_{\mathbb{N}_{r}}(1) = {'()'}$ is a word containing suffix  $'()'$, the word does not equal $s$.  The final possibility is  $\upgamma_{\mathbb{N}_{r}}(n)$ for some $n \ge 2$.  We note that regardless of which such $n$ we choose, its spelling only includes exponents $a_{i}$ from $a_{1}$ up to and including $a_{m}$, where $p_{m}$ is the greatest prime factor of $n$.   However, spelling function $\upgamma_{\mathbb{N}_{r}}$ only outputs an empty parenthesis pair if it encounters the number 1, either as the original number being spelled or in the course of spelling a prime number raised to the $0$th power.  But $m$ cannot be 0, since any prime number $p$ raised to the $0$th power is equal to one and is therefore not present in the prime factorization of $n$, implying $p$ cannot be the greatest prime factor of $n$.  This in turn implies that the spelling of $n$ cannot contain the suffix $'()'$, and we can thus be certain that $s \in \mathcal{D}$ but $s \notin \mathcal{D}_{r_{\text{min}}}$.
\end{itemize}
We therefore conclude that $\mathcal{D}_{r_{\text{min}}} \subsetneq \mathcal{D}$.
\end{proof}

We will see (Theorem~\ref{thm:DomainEvalDyck}) that every Dyck word represents a natural number. 
But the fact that $\mathcal{D}_{r_{\text{min}}}$ contains the spelling of every $n \in \mathbb{N}$, together with the fact that $\mathcal{D}_{r_{\text{min}}} \subset \mathcal{D},$ implies there must exist natural numbers possessing multiple representations in $\mathcal{D}$ (indeed, 0 is the \emph{only} natural number having a unique representation in $\mathcal{D}$, every other natural number having infinitely many).  Therefore we cannot view the number-words in $\mathcal{D}$ as being equivalent to the numbers they represent; to do so would result in absurdities such as the existence of more than one value of the product of two integers.

On the other hand, $\mathcal{D}_{r_{\text{min}}}$, enjoying a 1:1 correspondence with $\mathbb{N}$,  may be treated as if its members are the natural numbers themselves.  For instance, the equation
\begin{equation} \nonumber
6 \cdot 2^{3} + 2 = 50
\end{equation}
is equivalent to
\begin{equation} \nonumber
{(())(())} \cdot {(())}^{()(())} + {(())} = {(())()((()))},
\end{equation}
and the existence of a unique identity element for natural multiplication can be stated by asserting that for all $w \in \mathcal{D}_{r_{\text{min}}}$,
\begin{equation} \nonumber
{()} \cdot w = w = w \cdot {()}.
\end{equation}
In short, we may regard our $\mathcal{D}_{r_{\text{min}}}$ number-words as numbers.
And so I give $\mathcal{D}_{r_{\text{min}}}$ a simpler name:
\begin{definition} \label{def:DyckNaturals}
The \emph{set of Dyck natural numbers} is the language $\mathcal{D}_{r_{\text{min}}}$.
\end{definition}
\begin{remark}
Although we may indeed treat Dyck natural numbers as if they are the actual numbers themselves, relying upon the minimality of the natural interpretation, I recognize that a distinction exists between numbers and their representations in a numeral system.  Similar terminology is already widely used when discussing positional numeral systems, as when we speak of ``binary numbers'' or ``decimal numbers,'' despite the fact that a 1:1 correspondence does not exist between numbers and number words in such interpretations.
\end{remark}
\begin{remark}
There existing infinitely many $\upsigma$-permuted minimal natural RPF languages, I could choose any one of them to be the Dyck naturals.  But doing so would be no more advantageous than, say,  reversing the order of writing digits in decimal numbers so that $102_{10}$ would instead be written as $201_{01}$.  Thus I select the language underlying the standard interpretation.
\end{remark}

\subsection{The Standard RPF Natural Evaluation Functions $\boldsymbol{\upalpha_{\mathbb{N}_{r}}}$
and $\boldsymbol{\upalpha_{\mathbb{N}_{r_\text{min}}}}$} \label{loc:StdNatEvalFuncs}

Definition~\ref{def:NaturalSpellingFunction}  gave us the standard RPF natural spelling function $\upgamma_{\mathbb{N}_{r}}$, which is a bijection and therefore for which there exists an inverse function $\upgamma^{-1}_{\mathbb{N}_{r}}$.  But even though we are aware that the inverse exists, we do not yet know how to compute its values; we would like to correct this deficiency, especially since $\upgamma^{-1}_{\mathbb{N}_{r}}$ is the evaluator in $\text{RPF}_{\mathbb{N}_{r_{\text{min}}}}$.  An evaluation function mapping $\mathcal{D}_{r_{\text{min}}}$ onto $\mathbb{N}$ will thus be useful.
First, though, I will define a function having a domain equal to $\mathcal{D}$ rather than  $\mathcal{D}_{r_{\text{min}}}$, because I intend the function to eventually serve as the evaluator in the standard general RPF natural interpretation.  I will then restrict the function to   $\mathcal{D}_{r_{\text{min}}}$ to yield the evaluation
 function we previously referred to as $\upgamma^{-1}_{\mathbb{N}_{r}}$.
 
The definition of the evaluation function involves regarding the input word as a concatenation of chunks and recursively evaluating these.  Therefore before we go any further we must define exactly what we mean by a ``chunk,'' which in turn requires a definition of the ``dimensionality'' of a Dyck word.

\begin{definition} \label{def:DimensionalityOfADyckWord}
The \emph{dimensionality} of a Dyck word $w$ is the number of outermost matching parenthesis pairs in $w$.  More formally, the dimensionality of $w$ is given by the function $\text{dim}: \mathcal{D} \rightarrow \mathbb{N}$ such that
\begin{itemize}
\item For $w = \upepsilon$,  $\text{dim}(w) = 0$.
\item For $w \neq \upepsilon$,  $\text{dim}(w) = k$, where $k$ satisfies
\begin{equation}
 w{}={}\underset{i = 1}{\stackrel{k}{\bigoplus}} {(\;'('\myfrown{d_{i}}\myfrown')'\;)} \nonumber
\end{equation}
for some sequence of Dyck words $(d_{1},{}\ldots, d_{k})$.
\end{itemize}
\end{definition}
\begin{remark}
Dyck words $d_{1},{}\ldots, d_{k}$ are certain to exist, because a criterion for the well-formedness of a Dyck word is that if the word can be expressed as a string $s$ enclosed within a matching pair of parentheses, then $s$ is also a Dyck word.  Furthermore, the sequence $d$ is unique, since there is only one set of outermost matching parenthesis pairs in $w$, and $d_{i}$ is the Dyck word contained within the $i$th outermost matching parenthesis pair.
\end{remark}

\begin{definition} \label{def:Chunk}
A \emph{chunk} is a Dyck word of dimensionality 1.
\end{definition}

\begin{lemma} \label{lem:ChunkContainsDyckWordLenTwoLess}
If $w$ is a chunk, $w = $ $'('\myfrown d \myfrown ')'$ for some Dyck word $d$, and $|d| = |w|-2$.
\end{lemma}
\begin{proof}
Being a chunk, $w$ is a Dyck word of dimensionality 1.  Therefore
\begin{equation}
\begin{aligned}
 w{}&={}\underset{i = 1}{\stackrel{1}{\bigoplus}} {(\;'('\myfrown{d_{i}}\myfrown')'\;)}\\ \nonumber
       & =  {'('\myfrown{d_{1}}\myfrown')'} \\ \nonumber
       & =  {'('\myfrown{d}\myfrown')'}, \text{where } d = d_{1}. \\ \nonumber
\end{aligned}
\end{equation}
Since $w = {'('\myfrown{d}\myfrown')'}$, $d$ obviously has two fewer parentheses than $w$; thus $|d| = |w|-2$. 
\end{proof}

\begin{definition} \label{def:Content}
The \emph{content} of a chunk $w$ is the Dyck word $d$ satisfying the equation ${'('\myfrown{d}\myfrown')'} = w$.
\end{definition}

\begin{definition} \label{def:StdEvalDyck}
The \emph{standard RPF natural evaluation} of a Dyck word $w$ is the function $\upalpha_{\mathbb{N}_{r}} : \mathcal{D} \rightarrow \mathbb{N}$ such that
\begin{itemize}
\item For $w = \upepsilon$,  $\upalpha_{\mathbb{N}_{r}}(w) = 0$.
\item For $w \neq \upepsilon$,
\begin{equation}
\upalpha_{\mathbb{N}_{r}}(w) = \prod_{i = 1}^{\text{dim}(w)} {p_{i}^{\upalpha_{\mathbb{N}_{r}}(d_{i})}}  \nonumber
\end{equation}
where $d_{i}$ is the content of the $i$th chunk in $w$.
\end{itemize}
\end{definition}
\begin{remark}
The function is named $\upalpha$ for \begin{greektext}ἀριθμός\end{greektext} (\emph{arithmos}), Greek for ``number.''
\end{remark}
\begin{example} \label{exm:EvaluateSpellingOf27}
Let us evaluate {$()(()(()))$}:
\begin{equation} \nonumber
\begin{aligned}
&\upalpha_{\mathbb{N}_{r}}('()(()(()))') = \prod_{i = 1}^{2} p_{i}^{         \upalpha_{\mathbb{N}_{r}}(d_{i})} \\
& = p_{1}^{          \upalpha_{\mathbb{N}_{r}}(\upepsilon)                    }
p_{2}^{          \upalpha_{\mathbb{N}_{r}}('()(())')                    } 
  = p_{1}^{0}      p_{2}^{       p_{1}^{             \upalpha_{\mathbb{N}_{r}}(\upepsilon)              }                                
p_{2}^{          \upalpha_{\mathbb{N}_{r}}('()')                    }}  \\
& = p_{1}^{0}    p_{2}^{p_{1}^{0}      p_{2}^{        p_{1}^{            \upalpha_{\mathbb{N}_{r}}(\upepsilon)                  }                                  }                    }   
  = p_{1}^{0}  p_{2}^{p_{1}^{0}   p_{2}^{  p_{1}^{0}   }   } \\
 & = p_{2}^{p_{2}} = 3^{3} = 27.
\end{aligned}
\end{equation}
\end{example}
Because it may not be obvious to the reader that the standard natural evaluation function is defined for all $d \in \mathcal{D}$, I offer the following:
\begin{theorem} \label{thm:DomainEvalDyck}
The domain of the standard RPF natural evaluation function $\upalpha_{\mathbb{N}_{r}}$ is $\mathcal{D}$.
\end{theorem}
\begin{proof}
Let $S$ be the domain of $\upalpha_{\mathbb{N}_{r}}$.  We must show that $S \subseteq \mathcal{D}$ and $\mathcal{D} \subseteq{S}$.  The first of the two requirements is satisfied by the language of Definition~\ref{def:StdEvalDyck} itself, which explicitly restricts the domain of $\upalpha_{\mathbb{N}_{r}}$ to $\mathcal{D}$.   It therefore remains to be shown that  $\mathcal{D} \subseteq{S}$.  Toward that end, let $d$ be a member of $\mathcal{D}.$  Then either $d = \upepsilon$ or $d \neq \upepsilon$.
\begin{itemize}
\item For $d = \upepsilon$, $\upalpha_{\mathbb{N}_{r}}(d) $ is explicitly and uniquely defined to be $0$.
\item For $d \neq \upepsilon$, $\upalpha_{\mathbb{N}_{r}}(d) $ is recursively defined as $\prod_{i = 1}^{\text{dim}(d)} {p_{i}^{\upalpha_{\mathbb{N}_{r}}(e_{i})}}$,
where $e_{i}$ is the content of the $i$th chunk in d.  If $\upalpha_{\mathbb{N}_{r}}(e_{i})$ exists for all $i \in \{ {1,} {{ }\ldots{ }}, {\text{dim}(d)}\}$, then     $ \prod_{i = 1}^{\text{dim}(d)} {p_{i}^{\upalpha_{\mathbb{N}_{r}}(e_{i})}}$ exists as well, so that $\upalpha_{\mathbb{N}_{r}}(d)$ evaluates to a number.  Suppose that one of the terms $p_{i}^{\upalpha_{\mathbb{N}_{r}}(e_{i})}$ does not exist.  There are only two possible reasons for its nonexistence: either $e_{i}$ is not a Dyck word, or $\upalpha_{\mathbb{N}_{r}}(e_{i})$ involves an endless recursion of invocations and thus does not yield a value.  Let us show that neither of these supposed possibilities can be the case:
\begin{itemize}
\item
Lemma~\ref{lem:ChunkContainsDyckWordLenTwoLess} assures us that $e_{i}$, being the content of a chunk, must itself be a Dyck word.
\item
To see why the evaluation function cannot involve an endless recursion of invocations of itself, consider the evaluation of $()(()(()))$ from Example~\vref{exm:EvaluateSpellingOf27}.  
The recursion tree for $\upalpha_{\mathbb{N}_{r}}{({'}()(()(()))}{'})$ is:
\medskip
\begin{center}
\Tree [.$\upalpha_{\mathbb{N}_{r}}({'}()(()(())){'})$ $\upalpha_{\mathbb{N}_{r}}(\upepsilon)$ 
 [.$\upalpha_{\mathbb{N}_{r}}({'}()(()){'})$ $\upalpha_{\mathbb{N}_{r}}(\upepsilon)$ [.$\upalpha_{\mathbb{N}_{r}}({'}(){'})$ $\upalpha_{\mathbb{N}_{r}}(\upepsilon)$ ] ] ]
 \end{center}
 
 \medskip
 
The longest recursion chain arising from the evaluation of $()(()(()))$ is $(\upalpha_{\mathbb{N}_{r}}({'()(()(()))'}),
\upalpha_{\mathbb{N}_{r}}({'()(())'}),\upalpha_{\mathbb{N}_{r}}({'(())'}),\upalpha_{\mathbb{N}_{r}}({'()'}), \upalpha_{\mathbb{N}_{r}}(\upepsilon))$.  Now let $k$ be the greatest integer such that there is a recursion chain of length $k$ arising from the evaluation of an arbitrary Dyck word $d$.  Then $k$ must be finite, because the first term in the chain is the evaluation of a finite-length word and each subsequent term is the evaluation of the \emph{content} of a chunk from the evaluation of its predecessor, with the length of the content being 2 less than the length of the chunk from which it came, according to Lemma~\ref{lem:ChunkContainsDyckWordLenTwoLess}.   As the longest chain (or chains, if more than one chain is of length $k$) must be finite, the evaluation of $d$ cannot involve an infinite sequence of recursive invocations of $\upalpha_{\mathbb{N}_{r}}$ upon itself.
\end{itemize}
Since we have demonstrated that $\upalpha_{\mathbb{N}_{r}}(d)$ exists for every $d \in \mathcal{D}$, we have demonstrated that $\mathcal{D} \subseteq{S}$.
\end{itemize}

Having shown that $S \subseteq \mathcal{D}$ and $\mathcal{D} \subseteq{S}$, we conclude that $S = \mathcal{D}$.  Therefore the domain of $\upalpha_{\mathbb{N}_{r}}$ is $\mathcal{D}$.
\end{proof}

Note that $\upalpha_{\mathbb{N}_{r}}(w)$ exists for every $w \in \mathcal{D}_{r_{\text{min}}}$, since $\mathcal{D}_{r_{\text{min}}} \in \mathcal{D}$.   However, the inverse of bijective function $\upgamma_{\mathbb{N}_r}$ is \emph{not} $\upalpha_{\mathbb{N}_{r}}$, since the domain of the latter is a superset of $\mathcal{D}_{r_{\text{min}}}$.  We could opt to simply continue using $\upgamma^{-1}_{\mathbb{N}_r}$ to refer to the evaluation function in the standard minimal natural interpretation, using $\upalpha_{\mathbb{N}_{r}}$ to compute its values; nevertheless I prefer a more direct approach:
\begin{definition} \label{def:StandardMinimalNaturalEvaluationFunction}
The \emph{standard minimal RPF natural evaluation function}, denoted by $\upalpha_{\mathbb{N}_{r_{\text{min}}}}$, is the restriction of $\upalpha_{\mathbb{N}_{r}}$ to $\mathcal{D}_{r_{\text{min}}}$.
\end{definition}

I will make use of the following two lemmas to prove a theorem establishing that $\upalpha_{\mathbb{N}_{r_{\text{min}}}} = \upgamma^{-1}_{\mathbb{N}_r}$.
\begin{lemma} \label{lem:DimensionalityOfGammaIsGpf}
Let $n$ be a natural number greater than 1.  Then $\text{dim}(\upgamma_{\mathbb{N}_{r}}(n)) = m$, where $p_{m}$ is the greatest prime factor of $n$.
\end{lemma}
\begin{proof}
Referring to Definition~\ref{def:NonsurjectiveTranscriptionFunction}, which we use to calculate the value of $\upgamma_{\mathbb{N}_{r}}(n)$, we see that the spelling of $n$ results in a word containing $m$ outermost matching parenthesis pairs, where $p_{m}$ is the greatest prime factor of $n$.
\end{proof}

\begin{lemma} \label{lem:AlphaOfGammaOfNaturalEqualsNumberItself}
Let $p_{k}$ be the $k$th prime number.  Then $\upalpha_{\mathbb{N}_{r_{\text{min}}}}(\upgamma_{\mathbb{N}_r}(p_{k})) = p_{k}$.
\end{lemma}
\begin{proof}
Observe that
\begin{equation} \nonumber
\upgamma_{\mathbb{N}_r}(p_{k}) = \bigoplus_{i=1}^{k}( \;{'('} \myfrown \upgamma_{\mathbb{N}_r}(a_{i}) \myfrown {')'} \;),
\end{equation}
where every $a_{i}$ is zero except for $a_{k} = 1$.  We may thus rewrite $\upalpha_{\mathbb{N}_{r_{\text{min}}}}(\upgamma_{\mathbb{N}_r}(p_{k}))$ as
\begin{equation} \nonumber
\upalpha_{\mathbb{N}_{r_{\text{min}}}}(\bigoplus_{i=1}^{k}
( \;{'('} \myfrown \upgamma_{\mathbb{N}_r}(a_{i}) \myfrown {')'} \;)),
\end{equation}
 which in turn is equal to
 \begin{equation} \nonumber
 \prod^{k}_{i=1}p_{i}^{a_{i}}. 
 \end{equation}
Because the only nonzero $a_{i}$ is $a_{k} = 1$, we conclude that for every prime number $p_{k}$,
\begin{equation} \nonumber
\upalpha_{\mathbb{N}_{r_{\text{min}}}}(\upgamma_{\mathbb{N}_r}(p_{k})) = p_{k}.
\end{equation}
\end{proof}

With these lemmas in hand, I am ready to present the theorem and its proof.

\begin{theorem} \label{thm:StdNatEvalAndSpellingAreInverses}
The bijections  $\upalpha_{\mathbb{N}_{r_{\text{min}}}}$ and $\upgamma_{\mathbb{N}_r}$ are mutual inverses.
\end{theorem}
\begin{proof}
Our proof will be by mathematical induction. Let $P_{k}$ be the proposition
\begin{equation} \nonumber
\upalpha_{\mathbb{N}_{r_{\text{min}}}}(\upgamma_{\mathbb{N}_r}(k)) = k
\end{equation}
for natural number $k$.

Propositions $P_{0}$  and $P_{1}$ are easily verified by considering Definition~\ref{def:NonsurjectiveTranscriptionFunction} together with  Definition~\ref{def:StdEvalDyck}:
\begin{itemize}
\item $\upalpha_{\mathbb{N}_{r_{\text{min}}}}(\upgamma_{\mathbb{N}_r}(0)) =   \upalpha_{\mathbb{N}_{r_{\text{min}}}}(\upepsilon) = 0$.
\item $\upalpha_{\mathbb{N}_{r_{\text{min}}}}(\upgamma_{\mathbb{N}_r}(1)) = \upalpha_{\mathbb{N}_{r_{\text{min}}}}({'}(){'}) =  p_{1}^{0} = 1$.
\end{itemize}

Now let $n$ be a natural number such that $n \ge 2$ and $P_{l}$ is true for all $l$ less than $n$.

Choose nonnegative integers $a_{1}, \ldots, a_{m}$ satisfying
\begin{equation} \label{eqn:ChooseNonnegativeIntegersInProof}
\prod_{i=1}^{m}p_{i}^{a_{i}} = n,
\end{equation}
where $p_{m}$ is the greatest prime factor of $n$.  Then
\begin{equation} \nonumber
\upgamma_{\mathbb{N}_{r}}(n) = \bigoplus_{i=1}^{m}( \;'(' \myfrown\upgamma_{\mathbb{N}_{r}}(a_{i})\myfrown ')'\;).
\end{equation}

Applying $\upalpha_{\mathbb{N}_{r_{\text{min}}}}$ to both sides gives
\begin{equation} \nonumber
\upalpha_{\mathbb{N}_{r_{\text{min}}}}(\upgamma_{\mathbb{N}_{r}}(n))
= 
\upalpha_{\mathbb{N}_{r_{\text{min}}}}(\bigoplus_{i=1}^{m}( \;'(' \myfrown \upgamma_{\mathbb{N}_{r}}(a_{i})\myfrown ')'\;)).
\end{equation}

Observe that the dimensionality of $\upgamma_{\mathbb{N}_{r}}(n)$ is $m$, according to Lemma~\ref{lem:DimensionalityOfGammaIsGpf}; thus
\begin{equation} \label{eqn:DimensionalityInProof}
\upalpha_{\mathbb{N}_{r_{\text{min}}}}(\upgamma_{\mathbb{N}_{r}}(n)) = \prod_{i=1}^{m} p_{i}^{
\upalpha_{\mathbb{N}_{r_{\text{min}}}}(
\upgamma_{\mathbb{N}_{r}}(a_{i})
)}
.
\end{equation}

Note that either $n$ is a product of powers of prime numbers all of which are less than or equal to $n-1$, or $n$ is itself  prime.  But since  we are given that every natural number $l$ less than $n$ satisfies $P_{l}$, and since we know from Lemma~\ref{lem:AlphaOfGammaOfNaturalEqualsNumberItself} that  $\upalpha_{\mathbb{N}_{r_{\text{min}}}}(\upgamma_{\mathbb{N}_r}(p_{k})) = p_{k}$ for any prime number $p_{k}$, we can rewrite Equation~\ref{eqn:DimensionalityInProof} as
\begin{equation} \label{eqn:AlphaAsProdInProof}
\upalpha_{\mathbb{N}_{r_{\text{min}}}}(\upgamma_{\mathbb{N}_{r}(n)}) = \prod_{i=1}^{m}p_{i}^{a_{i}}.
\end{equation}

Considering Equation~\ref{eqn:ChooseNonnegativeIntegersInProof} together with  Equation~\ref{eqn:AlphaAsProdInProof}, we see that
\begin{equation} \nonumber
\upalpha_{\mathbb{N}_{r_{\text{min}}}}(\upgamma_{\mathbb{N}_r}(n)) = n.
\end{equation}

Thus $P_{k}$ is true for all $k \in \mathbb{N}$, implying that
$\upalpha_{\mathbb{N}_{r_{\text{min}}}}$
is the inverse of $\upgamma_{\mathbb{N}_{r}}$.
Furthermore, the two functions are mutual inverses, as the inverse of a bijection's  inverse is the bijection itself.
\end{proof}

Theorem~\ref{thm:StdNatEvalAndSpellingAreInverses} gives us certainty that $\upalpha_{\mathbb{N}_{r_{\text{min}}}}$ is the evaluation function in the standard minimal RPF natural interpretation.

\subsection{Prime Factorization of Dyck Naturals} \label{loc:PrimeFactorizationOfDyckNaturals}
Since words in the Dyck natural numeral system already recursively encode the prime factorizations of almost all  the numbers they represent, it is hardly surprising that their conventional factorization is straightforward, as I illustrate with Algorithm~\ref{alg:PrimeFactOfDyckNumber}.
{
\begin{algorithm}[htb]
 \KwInput{An arbitrary Dyck natural $w$ other than $\epsilon$ or $()$.}
 \KwOutput{The prime factorization of $w$ as an expression of Dyck naturals, with $\wedge$ as the exponentiation operator, square brackets as grouping symbols, and multiplication of grouped expressions implied by their juxtaposition.}
 \BlankLine
 \emph{base} $\leftarrow {{'}(()){'}}$\;
 \ForEach{chunk $c_{k}$ in $w$}{
  \If{$c_{k} \neq {'()'}$}{
   \textbf{print} $'['$\tcc*[r]{print without newline}
   \textbf{print} \emph{base}\;
   \textbf{print} ${'}\wedge{'}$\;
   \textbf{print} content of $c_{k}$\;
   \textbf{print} $']'$\;
   }
   \emph{base} $\leftarrow {'()'} \myfrown$\emph{base}\;
 }
 \caption{Finding the prime factorization of a Dyck natural.} \label{alg:PrimeFactOfDyckNumber}
\end{algorithm}
}
\begin{example}
 Given  ${(())()()(()(()))((()))}$ as input, the algorithm outputs the expression $[(())\wedge()][()()()(())\wedge()(())][()()()()(())\wedge(())]$.  Thus the algorithm tells us that
 \begin{equation} \nonumber
(())()()(()(()))((())) = (())^{()} \cdot ()()()(())^{()(())} \cdot ()()()()(())^{(())}.
 \end{equation}
 We can check the correctness of the equation by replacing Dyck naturals with their decimal equivalents, yielding
  \begin{equation} \nonumber
83006 = 2^{1} \cdot 7^{3} \cdot 11^{2}.
 \end{equation}
 \end{example}
 
 \renewcommand{\floatpagefraction}{0.7}
 \begin{remark}
 A reference implementation of the algorithm in Python may be found among the ancillary files accompanying this paper.
\end{remark}

Looking at the algorithm, we can see that there is only one loop, with the number of iterations being equal to the length of the input string.  But although the computations are straightforward, and while there is no recursion or backtracking, that does \emph{not} imply the algorithm is particularly efficient, as the time complexity of the algorithm is dominated by the output size.  For example, using 1 to encode a left parenthesis and 0 to encode a right parenthesis, the Dyck natural representation corresponding to the prime number  $p_{k}$ is matched by the regular expression
\begin{equation}
\text{\texttt{(10)\{k-1\}1100}}, 
\label{eq:my_regex_k1k2}
\end{equation}
implying the output for the factorization of the Dyck natural representing $p_{100} p_{200}$ would be
\begin{equation} \nonumber
[\underbrace{()\cdots ()}_{99 \text{ pairs}}(())\wedge()] [\underbrace{()\cdots ()}_{199 \text{ pairs}}(())\wedge()],
\end{equation}
for an output string length of 602, whereas using the decimal system we would have
\begin{equation} \nonumber
[541 \wedge 1] [1223 \wedge 1],
\end{equation}
giving an output string length of 15.

We might speed things up for our algorithm, however, by abbreviating words, e.g., by replacing sufficiently long strings of empty parenthesis pairs with hexadecimal words giving the lengths of unbroken sequences of such pairs; $(())()()()()()()()()()(())$ thus might be shortened to  $(())9(())$.  Under such a scheme the following two factorizations would be equivalent:
\begin{equation} \nonumber
\begin{aligned}
&(())()()()()()()()()()(()(())) = (())^{()} \cdot ()()()()()()()()()()(())^{()(())}\\
&(())9(()(())) = (())^{()} \cdot \text{A}(())^{()(())}.
\end{aligned}
\end{equation}

A prime factorization algorithm might use the compressed form for both input and output without ever having to reconstitute the corresponding Dyck words, since empty parenthesis pairs have no contents requiring processing.

\begin{remark}
The reader might object that finding the nonrecursive prime factorization of a Dyck natural number is a contrived problem, since the work involved with prime factorization will have already been performed ``up front'' in order to obtain the RPF representation in the first place.  But that is not necessarily true;  in Theorem~\vref{thm:EquivalentDefinitionsOfRPFNMin} we will see that $\mathcal{D}_{r_{\text{min}}}$ has an alternative nonnumerical definition, with the result that we can easily determine whether a given Dyck word is a Dyck natural, even without knowing the number it represents.  Furthermore, such representations may still have meaning, especially since related Catalan patterns appear naturally in structured or nested forms \cite{stanley2015catalan}.
\end{remark}

\subsection{Generalization to All Minimal Natural Interpretations} \label{loc:GeneralizationOfStdMinToAllPermutations}

So far I have confined my treatment of minimal interpretations of $\mathbb{N}$ to the ``standard'' or ``rightwise'' one; now I will provide a generalization to include all minimal prime-permuted natural interpretations.  Using this generalization, we will be able to convert objects in one permuted interpretation into corresponding objects in another.

 We will define a function $\upgamma'_{\mathbb{N}_{\upsigma}}$ mapping a natural number $n$ to a unique string of parentheses, and then we will define the language underlying $\text{RPF}_{\mathbb{N}_{\upsigma_{\text{min}}}}$ as the set of all possible strings produced by the function.  Since  $\upgamma'_{\mathbb{N}_{\upsigma}}$ takes a natural number input and outputs an $\text{RPF}_{\mathbb{N}_{\upsigma_{\text{min}}}}$ spelling of the number, we could choose to regard $\upgamma'_{\mathbb{N}_{\upsigma}}$ as our spelling function; however, because we wish the spelling function to have an inverse, we will instead define a spelling function $\upgamma_{\mathbb{N}_{\upsigma}}$ identical to $\upgamma'_{\mathbb{N}_{\upsigma}}$ except with its codomain restricted to $\text{RPF}_{\mathbb{N}_{\upsigma_{\text{min}}}}$.   Thus $\upgamma_{\mathbb{N}_{\upsigma}}$ will be a bijection, enabling us to speak of its inverse.

\begin{definition}\label{def:GeneralNonsurjectiveNaturalTranscriptionFunction}  Let $\Sigma^{*}$ be the Kleene closure of the set $\{(,)\}$, and let $\upsigma$ be a permutation of $\bm{P}$, where $\bm{P}$ is the sequence of prime numbers $(2,3,5,7, \ldots)$.  Then the \emph{$\upsigma$-permuted nonsurjective RPF natural transcription function}, denoted by   $\upgamma'_{\mathbb{N}_{\upsigma}}$, is given by $\upgamma'_{\mathbb{N}_{\upsigma}} : \mathbb{N} \rightarrow \Sigma^{*}$, such that
   \begin{itemize}
      \item For $ n = 0$, $\upgamma'_{\mathbb{N}_{\upsigma}}(n)$ is the empty string $\upepsilon$.
      \item For $n = 1$, $\upgamma'_{\mathbb{N}_{\upsigma}}(n)$ is the string ${()}$.
      \item For $n > 1$, let $(s_{k})_{k = 1}^{\infty}$ be the sequence $(\upsigma(p_{1}), \upsigma(p_{2}), \upsigma(p_{3}), \ldots)$, and let $m$ be the smallest number for which an integer sequence $(a_{j})_{j = 1}^{m}$ exists satisfying the equation
$n = \prod_{i = 1}^{m}s_{i}^{a_{i}}$.
Then
\begin{equation} \nonumber
\upgamma'_{\mathbb{N}_{\upsigma}}(n) = \bigoplus_{i = 1}^{m}('(' \myfrown \upgamma'_{\mathbb{N}_{\upsigma}}(a_{i}) \myfrown ')').
\end{equation}

   \end{itemize}
\end{definition}
\begin{equation} \nonumber  
{}
\end{equation}

At this point we can define our bijective spelling function by specifying its graph.

\begin{definition} \label{def:GeneralNaturalSpellingFunction}
\emph{The $\upsigma$-permuted RPF natural spelling function},
denoted by $\upgamma_{\mathbb{N}_{\upsigma}}$, is given by
\begin{equation}\nonumber
\upgamma_{\mathbb{N}_{\upsigma}} = \{(n,w) \in \mathbb{N} \times
\upgamma'_{\mathbb{N}_{\upsigma}}(\mathbb{N})
 \mid w = \upgamma'_{\mathbb{N}_{\upsigma}}(n)\},
\end{equation}
where $\upgamma'_{\mathbb{N}_{\upsigma}}(\mathbb{N})$ is the image of $\upgamma'_{\mathbb{N}_{\upsigma}}$.
\end{definition}

Before we define the $\upsigma$-permuted minimal RPF natural interpretation $\text{RPF}_{\mathbb{N}_{\upsigma_{\text{min}}}}$, we will assign a name to the language underlying the interpretation.

\begin{definition} \label{def:GenLDMin}
The \emph{$\upsigma$-permuted minimal RPF language}, denoted by $\mathcal{D}_{\upsigma_{\text{min}}}$, is the codomain of $\upgamma_{\mathbb{N}_{\upsigma}}$.
\end{definition}

Now we are ready to state the definition of $\text{RPF}_{\mathbb{N}_{\upsigma_{\text{min}}}}$.

\begin{definition}\label{def:GenMinimalRPFNaturalInterpretation}
   The \emph{$\upsigma$-permuted minimal RPF natural interpretation}, denoted by  $\text{RPF}_{\mathbb{N}_{\upsigma_{\text{min}}}}$, is the interpretation $(\mathcal{D}_{\upsigma_{\text{min}}}, \mathbb{N}, \upgamma^{-1}_{\mathbb{N}_{\upsigma}})$.
\end{definition}

An evaluation function mapping $\mathcal{D}_{\upsigma_{\text{min}}}$ into $\mathbb{N}$ will be useful.  However, I will define the function so that its domain is $\mathcal{D}$ rather than $\mathcal{D}_{\upsigma_{\text{min}}}$, as I intend for the definition to apply equally well to the general natural RPF language as to its subset $\mathcal{D}_{\upsigma_{\text{min}}}$.

\begin{definition} \label{def:GenEvalDyck}
Let  let $\upsigma$ be a permutation of the sequence $(2,3,5,7, \ldots)$ of prime numbers.  Then the \emph{$\upsigma$-permuted RPF natural evaluation} of a Dyck word $w$ is the function $\upalpha_{\mathbb{N}_{\upsigma}} : \mathcal{D} \rightarrow \mathbb{N}$ such that:
\begin{itemize}
\item For $w = \upepsilon$,  $\upalpha_{\mathbb{N}_{\upsigma}}(w) = 0$.
\item For $w \neq \upepsilon$, let $(s_{k})_{k = 1}^{\infty}$ be the sequence $(\upsigma(p_{1}), \upsigma(p_{2}), \upsigma(p_{3}), \ldots)$.  Then
\begin{equation}
\upalpha_{\mathbb{N}_{\upsigma}}(w) = \prod_{i = 1}^{\text{dim}(w)} {
\upsigma(p_{i})
^{\upalpha_{\mathbb{N}_{\upsigma}}(d_{i})}},  \nonumber
\end{equation}
where $d_{i}$ is the content of the $i$th chunk in $w$.
\end{itemize}
\end{definition}

Note that $\upalpha_{\mathbb{N}_{\upsigma}}(w)$ exists for every $w \in \mathcal{D}_{\upsigma_{\text{min}}}$, since $\mathcal{D}_{\upsigma_{\text{min}}} \in \mathcal{D}$.   However, the inverse of bijective function $\upgamma_{\mathbb{N}_\upsigma}$ is \emph{not} $\upalpha_{\mathbb{N}_{\upsigma}}$, the domain of the latter being a superset of $\mathcal{D}_{\upsigma_{\text{min}}}$.  The following definition gives us a $\upsigma$-permuted evaluation function that is the inverse of the $\upsigma$-permuted spelling function:
\begin{definition} \label{def:SigmaPermutedMinimalNaturalEvaluationFunction}
The \emph{$\upsigma$-permuted minimal RPF natural evaluation function}, denoted by $\upalpha_{\mathbb{N}_{\upsigma_{\text{min}}}}$, is the restriction of $\upalpha_{\mathbb{N}_{\upsigma}}$ to $\mathcal{D}_{\upsigma_{\text{min}}}$.
\end{definition}

Now we have the ability to convert spellings in one miminal natural interpretation to equivalent spellings in another, according to the following straightforward procedure.  Suppose we have a word $w_{\upsigma}$ evaluating to some natural number $n$ in the $\upsigma$-permuted minimal natural interpretation and wish to find the word $w_{\uptau}$ evaluating to $n$ in the $\uptau$-permuted minimal natural interpretation.  We first find $n$ by applying $\upalpha_{\mathbb{N}_{\upsigma_{\text{min}}}}$ to $w_{\upsigma}$; we then apply $\upgamma_{\mathbb{N}_{\uptau}}$ to $n$, giving us $w_{\uptau}$.  In other words, we use the following equation:
\begin{equation} \label{eqn:ConvertSpellingsBetweenPermutations}
w_{\uptau} = \upgamma_{\mathbb{N}_{\uptau}}(\upalpha_{\mathbb{N}_{\upsigma_{\text{min}}}}(w_{\upsigma})).
\end{equation}

We can also convert the evaluations in one minimal natural interpretation to equivalent evaluations in another, as follows.  Suppose we have a number $n_{\upsigma} \in \mathbb{N}$, the spelling of which is $w$ in  the $\upsigma$-permuted minimal natural interpretation, and we wish to find the number $n_{\uptau} \in \mathbb{N}$ with the spelling $w$ in the $\uptau$-permuted minimal natural interpretation.  The following equation gives us $n_{\uptau}$:
\begin{equation} \label{eqn:ConvertEvaluationsBetweenPermutations}
n_{\uptau} = \upalpha_{\mathbb{N}_{\uptau_{\text{min}}}}(
\upgamma_{\mathbb{N}_{\upsigma}}(n_{\upsigma})
).
\end{equation}

For the rest of this paper, I will focus primarily upon interpretations arising from the standard permutation. Generalizations of the associated mathematical objects to all prime-permuted interpretations are somewhat tedious but straightforward.

\subsection{The Standard Minimal RPF Language Reconsidered} \label{loc:StdMinRPFLangReconsidered}

We defined $\mathcal{D}_{r_{\text{min}}}$ to be the codomain of $\upgamma_{\mathbb{N}_{r}}$  (Definition~\ref{def:LDMin}).  However, there is an interesting alternative definition---I say interesting, 
because it is a non-numerical definition of $\mathcal{D}_{r_{\text{min}}}$, and because it can be used to prove that $\mathcal{D}_{r_{\text{min}}}$ is a context-free language (see Theorem~\vref{thm:RPFChomskySch}).

\begin{theorem} \label{thm:EquivalentDefinitionsOfRPFNMin}
The following are equivalent:
\begin{enumerate}
   \item   The standard minimal RPF language $\mathcal{D}_{r_{\text{min}}}$ is the domain of the standard minimal RPF natural evaluation function $\upalpha_{\mathbb{N}_{r_{\text{min}}}}$.
   \item   The standard minimal RPF language $\mathcal{D}_{r_{\text{min}}}$ is the codomain of the standard RPF natural spelling function $\upgamma_{\mathbb{N}_{r}}$. \label{prt:1:thm:EquivalentDefinitionsOfRPFNMin}
\item  The standard minimal RPF language $\mathcal{D}_{r_{\text{min}}}$ is the set
\begin{equation} \nonumber
\{ d \in \mathcal{D} \mid (\;')())' \; \text{is not a substring of}\;d \; ) {\;\wedge\;} (\; ')()' \; \text{is not a suffix of} \; d\;)\}.
\end{equation} \label{prt:2:thm:EquivalentDefinitionsOfRPFNMin}
\end{enumerate}
\end{theorem}
\begin{proof}
Theorem~\vref{thm:StdNatEvalAndSpellingAreInverses} establishes the equivalence of (1) and (2).   Demonstrating the equivalence of either of these with (3) would therefore prove the equivalence of all three statements; I shall proceed by demonstrating the equivalence of (2) and (3).

The spellings of $0$ and $1$ are $\upepsilon$ and ${()}$ respectively, by explict definition; neither of these contains the substring  {$)())$} or the suffix {$)()$}.

Suppose that there exists some natural number $n$ greater than 1 such that its spelling contains the suffix {$)()$}, i.e, $\upgamma_{\mathbb{N}_{r}}(n) = w = {w'} \myfrown {'()'}$, with $w'$ being a nonempty Dyck word.  This would imply that the last chunk in $w$ corresponds to the zeroeth power of the greatest prime factor of $n$, which is a contradiction, the zeroeth power of a prime number not appearing in the (nonrecursive) prime factorization of $n$.

Now suppose that there exists some natural number $m$ greater than 1 such that its spelling contains the substring {$)())$}.  The presence of the closing parenthesis immediately following the empty matched pair of parentheses would imply that $\upgamma_{\mathbb{N}_{r}}$ had been invoked with some number $n \ge 2$ as its argument as part of the recursion arising from the spelling of some other number, and that the spelling of $n$ contained the suffix {$)()$}, which we have already shown is a contradiction.

As for all other words $w$ in $\mathcal{D}$ other than those we have already considered, these must also be members of $\mathcal{D}_{r_{\text{min}}}$.  To see this, suppose that $w$ is not in the codomain of $\upgamma_{\mathbb{N}_{r}}$.  Since the domain of $\upalpha_{\mathbb{N}_{r}}$ is $\mathcal{D}$, $\upalpha_{\mathbb{N}_{r}}(w)$ is a natural number.  This would imply that a natural number exists (other than 0 or 1, these already having been considered) which cannot be represented by a product of prime powers
\begin{equation} \nonumber
n = \prod_{i = 1}^{m}p_{i}^{a_{i}},
\end{equation}
where $p_{m}$ is the greatest prime factor of $n$ and $a$ is the unique integer sequence satisfying the equation.  Note that if we form a product of only the $p^{a_{i}}$ where $a_{i} \neq 0$, we obtain the (nonrecursive) prime factorization of $n$.  Also note that all zero $a_{i}$ designate the exponents of primes not contributing to the prime factorization.  Thus $n$ is an integer greater than or equal to 2 such that $n$ has no prime factorization, which is a contradiction.
\end{proof}

We now make use of the alternative definition of $\mathcal{D}_{r_{\text{min}}}$ to prove that $\mathcal{D}_{r_{\text{min}}}$ is context-free.

\begin{theorem} \label{thm:RPFChomskySch}
Let $\Sigma^{*}$ be the Kleene closure of the set $\{(,)\}$,
and let $R$ be the language accepted by the deterministic finite automaton (DFA) represented by the following state transition diagram:

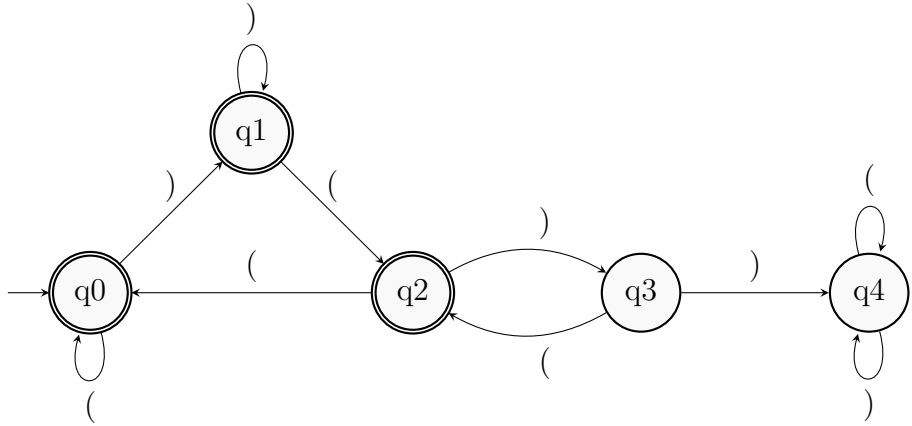
\begin{center}
\tikzset{
->,
>=stealth, 
node distance=3cm,
every state/.style={thick, fill=gray!5},
initial text=$ $,
}
\begin{tikzpicture}
   \node[state,initial,accepting](q0){q0};
   \node[state,accepting,above right of=q0](q1){q1};
   \node[state,accepting,below right of=q1](q2){q2};
   \node[state,right of=q2](q3){q3};
   \node[state,right of=q3](q4){q4};

   \draw (q0) edge[loop below]node{(}(q0)
             (q0)edge[above]node{)}(q1)
             (q1)edge[loop above]node{)}(q1)
             (q1)edge[above]node{(}(q2)
             (q2)edge[above]node{(}(q0)
             (q2)edge[above, bend left]node{\;\;\;\;)}(q3)
             (q3)edge[below, bend left]node{\;\;\;\;(}(q2)
             (q3)edge[above]node{)}(q4)
             (q4)edge[loop above]node{(}(q4)
             (q4)edge[loop below]node{)}(q4);
\end{tikzpicture}
\end{center}
Then
\begin{equation} \nonumber
R =
\{w \in \Sigma^{*} \mid (\text{${')())'}$ not a substring of $w$})\wedge(\text{${')()'}$ not a suffix of $w$})\},
\end{equation}
and
\begin{equation} \nonumber
\mathcal{D}_{r_{\text{min}}} = \mathcal{D} \cap R,
\end{equation}
implying that $\mathcal{D}_{r_{\text{min}}}$ is a context-free language.
\end{theorem}

\begin{proof}
We first verify that no word $w \in \Sigma^{*}$ containing the substring ${)())}$ is accepted by the DFA.  Suppose the automaton has read some arbitrary number of symbols and is currently at state $q_{i}$, with the remaining input starting with the string ${)())}$.
For each $i \in \{0,1,2,3,4\}$, the input sequence ${)())}$ places the automaton in state $q_{4}$, from which no transition to another state is possible, so that $w$ ends with the automaton in state $q_{4}$.  Since $q_{4}$ is not an acceptor state, $w$ is rejected.  Thus the DFA rejects all words containing the substring ${)())}$.  

Next we verify that no word $w$ in $\Sigma^{*}$ containing the suffix ${)()}$ is accepted by the DFA.  Suppose the automaton has read some arbitrary number of symbols and is currently at state $q_{i}$, the remaining input being ${)()}$.  Now let us consider each of the possibilities.

\begin{itemize}
\item If $i \in \{0,1,2\}$, then the input ends with the automaton in state $q_{3}$.  Since $q_{3}$ is not an acceptor state, $w$ is rejected.
\item If $i \in \{3,4\}$, then the input ends with the automaton in state $q_{4}$.  Since $q_{4}$ is not an acceptor state, $w$ is rejected.
\end{itemize}
Thus the automaton rejects all words in $\Sigma^{*}$ containing the suffix ${)()}$.

\vspace{\baselineskip} 
We must still verify that all words $w$ in $\Sigma^{*}$ other than those rejected above are recognized by the DFA.  In our verification, we can ignore any words ending with the DFA in state $q_{4}$, since the only way to reach that state is for $w$ to contain the substring ${)())}$, and all of these words have already been considered.  We may likewise ignore words ending with the DFA in state $q_{3}$, these all containing the suffix ${)()}$.  The only other possibility is that $w$ ends in either $q_{0}$, $q_{1}$ or $q_{2}$, each of which is an acceptor state.  Thus the automaton accepts all words in $\Sigma^{*}$ except for those containing the substring ${)())}$ or the suffix $)()$.

Because $\mathcal{D} \subset \Sigma^{*}$, the DFA accepts every $d \in \mathcal{D}$ such that ${)())}$ is not a substring of $d$ and ${)()}$ is not a suffix of $d$.  Therefore, from  Theorem~\ref{thm:EquivalentDefinitionsOfRPFNMin}
,
\begin{equation} \nonumber
\mathcal{D}_{r_{\text{min}}} = \mathcal{D} \cap R.
\end{equation}
Thus $\mathcal{D}_{r_{\text{min}}}$ is a regular subset of a context-free language.  Since all regular languages are also context-free, we conclude that $\mathcal{D}_{r_\text{min}}$ is a context-free language.

\end{proof}

The non-numeric interpretation of $\mathcal{D}_{r_\text{min}}$ implies that the number of words of semilength $k$ in that set is given by OEIS sequence A082582  \cite{oeisA082582}  by the following theorem:

\begin{theorem} \label{thm:A082582}
Let $k$ be a nonnegative integer.  Then the number of words in $\mathcal{D}_{r_{\text{min}}}$ of length
$2k$ is given by term a($k$) in OEIS sequence A082582.
\end{theorem}

\begin{proof}
The recurrence and associated generating function establishing this result were contributed by Max Alexeyev in his answer to a MathOverflow question \cite{mo384586}.
\end{proof}

\section{The Standard Minimal RPF Superrational Interpretation $\text{RPF}_{\mathbb{S}_{r_\text{min}}}$} \label{loc:StdMinRpfRationalInterpretation}

I extend the Dyck natural numeral system to yield a system capable of representing all members of $\mathbb{Q}$, as well as some algebraic irrationals and even some transcendental numbers.

\subsubsection{An Algebraic Definition of the Superrational Numbers} \label{loc:AlgebraicDefOfS}
In order to introduce the set of superrational numbers, which will be the target set in the Dyck superrational numeral system, I first begin with an alternative construction of the natural numbers inspired by the fundamental theorem of arithmetic; I will then go on to extend the construction, thereby defining the set of superrational numbers.

\begin{definition} \label{def:Natural}
The set $\mathbb{N}$ of nonnegative integers, also called the set of \emph{natural numbers}, is the smallest subset  of the set $\mathbb{R}$ of real numbers satisfying:
\begin{enumerate}
  \item $0 \in \mathbb{N}$. (Existence of a zero element)
  \item If $p$ is prime and $n \in \mathbb{N}$, then $p^{n} \in \mathbb{N}$. (Prime-base exponentiation)
  \item If $m,n \in \mathbb{N}$, then $mn \in \mathbb{N}$.  (Multiplication)
\end{enumerate}

\end{definition}

 For convenience in the discussion which follows, I use the following terminology:
\begin{definition} \label{def:PowerTowerOfPrimes}
A \emph{power tower of primes} is a power tower in which the base and all exponents are prime numbers.
\end{definition}
\begin{definition}  \label{def:PowerTowerOfPrimesWithNegatableExponents}
A \emph{power tower of primes with negatable exponents} is a power tower in which the base is a prime number and the absolute values of all exponents are prime numbers.
\end{definition}

It is straightforward to show that  Definition~\ref{def:Natural} yields exactly the set of
nonnegative integers:

\begin{itemize}
  \item By closure rules (1) and (2), the zeroth powers $p^{0}$ of all prime numbers $p$ are in
  $\mathbb{N}$. Since $p^{0} = 1$ for all primes $p$, we know that $1 \in \mathbb{N}$.

  \item With $1 \in \mathbb{N}$, closure rule (2) implies that the first powers
  $p^{1}$ of all prime numbers $p$ are also in $\mathbb{N}$. Thus all prime numbers
  are in $\mathbb{N}$.

  \item Again applying closure rule (2), we see that all integers expressible as finite power towers of primes lie in $\mathbb{N}$.

  \item Applying closure rule (3), it follows that all integers expressible as products of finite power towers of primes are in $\mathbb{N}$. By the fundamental theorem of arithmetic, every integer $n \ge 2$ admits such a representation and therefore is a member of $\mathbb{N}$.  
\end{itemize}
And so we see that the set $\mathbb{N}$ of nonnegative integers as constructed in Definition~\ref{def:Natural} is the smallest subset of $\mathbb{R}$
  containing $0$, $1$, and all integers greater than $1$.

\begin{remark}
An analogous alternative definition of the set $\mathbb{N}_+$ of positive integers is the closure of $\{1\}$ under prime-base exponentiation and multiplication.
\end{remark}

The inclusion of just one more operation yields a vastly larger set.

\begin{definition} \label{def:Superrational}
The set $\mathbb{S}$, also called the set of \emph{superrational numbers}, is the smallest subset of $\mathbb{R}$ satisfying:
\begin{enumerate}
  \item $0 \in \mathbb{S}$. (Existence of a zero element)
  \item If $p$ is prime and $s \in \mathbb{S}$, then $p^{s} \in \mathbb{S}$. (Prime-base exponentiation)
  \item If $s_1,s_2 \in \mathbb{S}$, then $s_1 s_2 \in \mathbb{S}$.  (Multiplication)
  \item If $s \in \mathbb{S}$, then $-s \in \mathbb{S}$.  (Negation)
\end{enumerate}

\end{definition}

The properties already present from our treatment of $\mathbb{N}$ imply that
$\mathbb{N} \subseteq \mathbb{S}$. Let us now consider some consequences arising
from including the operation of negation in the closure:

\begin{itemize}
  \item Since all natural numbers are in $\mathbb{S}$, all of their negatives
  are in $\mathbb{S}$ as well. Hence all integers are in $\mathbb{S}$. Thus,
  for example, $-12 \in \mathbb{S}$.

  \item Since all negative integers are in $\mathbb{S}$, (2) and (3) taken
  together imply that all rational numbers are in $\mathbb{S}$. Thus, for
  example, $2^{-1} = \frac{1}{2}$ and
  $2^{5}13^{-2} = \frac{32}{169}$ are in $\mathbb{S}$.

  \item Since all rational numbers are in $\mathbb{S}$, (2) and (3) further
  imply that some algebraic irrational numbers, such as
  $2^{2^{-1}} = 2^{\frac{1}{2}} = \sqrt{2}$ and
  $2^{2^{-1}}3^{2^{-1}} = 2^{\frac{1}{2}}3^{\frac{1}{2}} = \sqrt{6}$, are in
  $\mathbb{S}$.

  \item Since $\sqrt{2}$ is in $\mathbb{S}$, (2) implies that the
  Gelfond--Schneider constant $2^{\sqrt{2}}$, a transcendental number, is in
  $\mathbb{S}$.
  
  \item Taken together, the four closure rules imply that $\mathbb{S}$ equals the smallest subset of $\mathbb{R}$ containing 0, 1, -1, and all products of finite power towers of primes with negatable exponents.
\end{itemize}

And so we see that the recursive nature of these set definitions implies that merely
extending the closure of $\{0\}$ under multiplication and prime-base
exponentiation to include an additional operation of negation expands the set
to contain all rational numbers, a subset of the
algebraic irrationals, and even a subset of the transcendental numbers.  Furthermore, we identify $\mathbb{S}$ as the smallest subset of $\mathbb{R}$ containing $0, \pm 1$, and all numbers expressible as products of finite power towers of primes with negatable exponents.

\begin{remark}
The algebraic development in this section, to include the alternative definition of the natural numbers and then its extension to define the superrational numbers, would never have occurred to me had I not been looking at topics in number theory from the perspective of formal language theory.  What I am doing might arguably be called \emph{grammatical number theory}.
\end{remark}

\subsection{Extension of Prime Factorization to Nonzero Rational Numbers}  \label{loc:Rational_prime_factorizations}

This section may be safely skipped, being only intended to facilitate an intuitive understanding of superrational recursive prime factorizations by approaching them from a different perspective.

I describe an extension of nonrecursive integer prime factorization such that every rational number other than 0, 1 and -1 may be viewed as if it has a prime factorization.  

The prime factorization of of each natural number $n \ge 2$ is a product of its prime factors such that the product evaluates to $n$ (if we understand prime numbers to be their own prime factorizations); we are guaranteed by the fundamental theorem of arithmetic that the product is unique up to the order of the factors.  Thus we may write $40$ as
$2 \times 2 \times 2 \times 5$, or more briefly in exponential form as $2^3 \cdot 5^1$.  Let us consider the general exponential form of the prime factorization of $n$
\begin{equation} \nonumber
p_{a_1}^{b_1} \cdots p_{a_k}^{b_k},
\end{equation}
where $p_{a_{i}}$ is the $i$th prime factor of $n$, $p_{a_{k}}$ is the greatest prime factor of $n$, and $(b_{1}, \ldots b_{k})$ is the unique sequence of positive integers such that the expression evaluates to $n$.

But let us now relax our requirement that the integers $b_{i}$ in $(b_{i})_{i=1}^{k}$ be positive, allowing negative integers to be included as well. Then we could write $5.6$, for example, as
\begin{equation} \nonumber
p_{1}^{2} \cdot p_{3}^{-1} \cdot p_{4}^{1}.
\end{equation}
Observe that the above expression bears a striking resemblance to the exponential form of prime factorization. Indeed, just as with prime factorization, the expression is unique up to the order of its factors, since $5.6$ can be uniquely expressed as the reduced fraction $\frac{28}{5}$, the numerator and denominator each corresponding to unique prime factorizations:
\begin{equation} \nonumber
5.6 = \frac{28}{5} = \frac{2\times 2 \times 7}{5}.
\end{equation}
If we are willing to include -1 as a factor (though freely admitting that -1 is not  prime), we may represent negative rationals as well.  Hence we may extend prime factorization to rational numbers other than -1, 0, and 1.
\begin{definition} \label{def:PrimeFactorizationOfRationalNumber}
Let $q$ be a rational number other than -1, 0 or 1, and let $(b_{1}, \ldots, b_{k})$ be the integer sequence of shortest length such that ${\mid}q{\mid}$ is equal to $p_{a_{1}}^{b_{1}} \cdots p_{a_{k}}^{b_{k}}$, where $(p_{a_{i}})_{i=1}^{k}$ is a subsequence of the sequence $\mathbf{P}$ of prime numbers.  Then the \emph{rational prime factorization} of $q$ is 
$p_{a_{1}}^{b_{1}} \cdots p_{a_{k}}^{b_{k}}$ if $q > 0$; otherwise, the rational prime factorization of $q$ is
$-1 \cdot p_{a_{1}}^{b_{1}} \cdots p_{a_{k}}^{b_{k}}$.
\end{definition}
\begin{remark}
I used the name ``\emph{rational} prime factorization'' rather simply ``prime factorization'' in order to emphasize that the extension is just that, an extension.   The concept of prime factorization is indeed distinct from that of \emph{rational} prime factorization, as the fundamental theorem of arithmetic only concerns integers, specifically those greater than 1.  Also, note that while I was able to write the prime factorization of decimal 40 without using exponentiation, there is no integer $k$ satisfying
\begin{equation} \nonumber
5.6 = 2 \times 2 \times \underbrace{5 \times \cdots \times 5}_{k \text{ $5$s}} \times 7.
\end{equation}
\end{remark}

\subsection{The Standard Superrational Spelling Function $\upgamma_{\mathbb{S}_{r}}$}
\begin{definition} \label{def:MostSignificantPrimeBase}
Let $s$ be a nonzero superrational number.  The \emph{greatest prime base} of $s$ is given by the function $
\text{gpb}: {\mathbb{S} \setminus \{0\}} \rightarrow \{p_{1}, p_{2}, p_{3}, \ldots\}$ such that
\begin{itemize}
\item
If $|s| = 1$, then $\text{gpb}(s) = p_{1}$.
\item
Otherwise, let the integer sequence $(a_{1}, \ldots, a_{m})$ satisfy the equation  $s = \prod_{i = 1}^{m} p_{i}^{a_{i}}$ such that $m$ is the greatest number for which $a_{m} \neq 0$.  Then  $\text{gpb}(s) = p_{m}$.
\end{itemize}
\end{definition}
\begin{remark}
Informally, the greatest prime base may be thought of as the largest prime number that must appear in the product in order for the product to evaluate to $s$.
\end{remark}

As before, but with respect to $\mathbb{S}$, we will first define a nonsurjective precursor to our bijective spelling function and then define the spelling function by specifying its graph.
\begin{definition} \label{def:Nonsurjective_Superrational_Spelling}  Let $\Sigma^{*}$ be the Kleene closure of the set $\{(,)\}$.  Then the \emph{standard nonsurjective RPF superrational transcription function}, denoted by   $\upgamma'_{\mathbb{S}_{r}}$, is given by  $\upgamma'_{\mathbb{S}_{r}} : \mathbb{S} \rightarrow \Sigma^{*}$, where
\begin{itemize}
\item
For $s = 0$, $\gSr[s][\prime] = \upepsilon$.
\item
For $s = 1$, $\gSr[s][\prime] = {'()'}$.
\item
For $s < 0$,  $\gSr[s][\prime] = \gSr[|s|][\prime]\myfrown {'()'}$.
\item
Otherwise, let $p_{m}$ be the greatest prime base of $s$, and let integer sequence $(a_{1}, \ldots, a_{m})$ satisfy the equation  $s = \prod_{i = 1}^{m} p_{i}^{a_{i}}$.  Then
\begin{equation}
            \gSr[s][\prime] =
            { } \underset{i = 1}{\stackrel{m}{\bigoplus}} {(\;{'('}\myfrown
            \upgamma_{\mathbb{S}\textsubscript{r}}^{\prime}(a_{i}})       
            \myfrown{')'\;)}. \nonumber
\end{equation}
\end{itemize}
\end{definition}

\begin{definition} \label{def:SuperrationalSpellingFunction}
\emph{The standard RPF superrational spelling function},
also called  ``the standard minimal RPF superrational spelling function'' and
denoted by $\upgamma_{\mathbb{S}_{r}}$, is given by
\begin{equation}\nonumber
\upgamma_{\mathbb{S}_{r}} = \{(s,w) \in \mathbb{S} \times
\upgamma'_{\mathbb{S}_{r}}(\mathbb{S})
 \mid w = \upgamma'_{\mathbb{S}_{r}}(s)\},
\end{equation}
where $\upgamma'_{\mathbb{S}_{r}}(\mathbb{S})$ is the image of $\upgamma'_{\mathbb{S}_{r}}$.
\end{definition}

\begin{remark}
A partial reference implementation of Definition~\ref{def:SuperrationalSpellingFunction} may be found among the ancillary files accompanying this paper.
\end{remark}

\begin{example} \label{ex:SpellingOfNegativeTwoNinths}
Let us spell the Dyck superrational corresponding to decimal $-0.\overline{2}$.

First we note that  $-0.\overline{2} = -\frac{2}{9} = -1 \cdot 2^{1} \cdot 3^{-2} = -p_{1}^{1}p_{2}^{-2}$.
Thus the standard RPF superrational spelling corresponding to  decimal $-0.\overline{2}$ is

\begin{equation} \nonumber
\begin{aligned}
&\upgamma_{\mathbb{S}_{r}}(-0.\overline{2})  = \upgamma_{\mathbb{S}_{r}}(0.\overline{2}) \myfrown {'()'}    \\
 & = {'('} \myfrown 
 \upgamma_{\mathbb{S}_{r}}(1) \myfrown
 {')'} \myfrown
 {'('} \myfrown
 \upgamma_{\mathbb{S}_{r}}(-2) \myfrown
 {')'} \myfrown {'()'} 
 \\
 & = {'(())'} \myfrown
 {'('} \myfrown
 \upgamma_{\mathbb{S}_{r}}(2) \myfrown {'()'}
 \myfrown {')()'} 
 \\
  & = {'(())('} \myfrown
 {'('} \myfrown 
 \upgamma_{\mathbb{S}_{r}}(1) \myfrown
 {')'} \myfrown {'())()'} 
 \\
  & = {(())((())())()}.
\end{aligned}
\end{equation}
\end{example}

\begin{example} \label{ex:SpellingOfSqrtTwo}
Having spelled a Dyck superrational corresponding to a rational number, let us now spell the Dyck superrational  corresponding to $\sqrt{2}$, an algebraic irrational number.

Observe that  $\sqrt{2} = 2^{\frac{1}{2}} = 2^{2^{-1}}$.
Thus the standard RPF superrational spelling corresponding to $\sqrt{2}$ is

\begin{equation} \nonumber
\begin{aligned}
&\upgamma_{\mathbb{S}_{r}}(\sqrt{2})  =  \upgamma_{\mathbb{S}_{r}}(2^{2^{-1}}) \\
 & = {'('} \myfrown  \upgamma_{\mathbb{S}_{r}}(2^{-1}) \myfrown   {')'} \\
 & = {'('} \myfrown      {'('} \myfrown       \upgamma_{\mathbb{S}_{r}}(-1)     \myfrown   {')'}    \myfrown   {')'}  \\
 & = {'('} \myfrown      {'('} \myfrown       \upgamma_{\mathbb{S}_{r}}(1) \myfrown   {'()'}   \myfrown   {')'}    \myfrown   {')'}  \\
 & = {'('} \myfrown      {'('} \myfrown   {'()'}    \myfrown   {'()'}   \myfrown   {')'}    \myfrown   {')'}  \\
 & = {((()()))}.
\end{aligned}
\end{equation}
\end{example}
\begin{remark} \label{ex:SpellingOfSqrt2}
The OEIS sequence entitled ``Expansion of sqrt(2) in base 2'' \cite{oeisA004539}, which is infinite, thus corresponds to the finite sequence $(1,1,1,0,1,0,0,0)$, which I show with left parentheses encoded by 1s and right parentheses encoded by 0s in order to avoid confusion with parentheses as used in sequence notation.  
\end{remark}

\begin{example} \label{ex:SpellingOfGelfondSchneider}
Let us spell the Dyck superrational corresponding to the Gelfond-Schneider constant, the transcendental number $2^{\sqrt{2}}$.

Since  $2^{\sqrt{2}} = 2^{2^{2^{-1}}}$,
the standard RPF superrational spelling corresponding to $2^{\sqrt{2}}$ is

\begin{equation} \nonumber
\begin{aligned}
&\upgamma_{\mathbb{S}_{r}}(2^{\sqrt{2}})  =  \upgamma_{\mathbb{S}_{r}}(2^{2^{2^{-1}}}) \\
 & = {'('} \myfrown  \upgamma_{\mathbb{S}_{r}}(2^{2^{-1}}) \myfrown   {')'} \\
 & = {'('} \myfrown {'('} \myfrown  \upgamma_{\mathbb{S}_{r}}(2^{-1}) \myfrown   {')'} \myfrown   {')'} \\ 
 & = {'('} \myfrown    {'('} \myfrown   {'('} \myfrown       \upgamma_{\mathbb{S}_{r}}(-1)     \myfrown   {')'}  \myfrown   {')'}   \myfrown   {')'}  \\
 & = {'('} \myfrown    {'('} \myfrown   {'('} \myfrown       \upgamma_{\mathbb{S}_{r}}(1) \myfrown   {'()'}   \myfrown   {')'}    \myfrown   {')'} \myfrown   {')'}  \\
 & = {'('} \myfrown   {'('} \myfrown    {'('} \myfrown   {'()'}    \myfrown   {'()'}   \myfrown   {')'}    \myfrown   {')'} \myfrown   {')'} \\
 & = {(((()())))}.
\end{aligned}
\end{equation}
\end{example}
\begin{remark} \label{ex:SpellingOf2ToSqrt2}
The OEIS sequence entitled ``Decimal expansion of 2\^{}sqrt(2)'' \cite{oeisA007507}, which is inifinite, thus corresponds to the finite sequence $(1,1,1,1,0,1,0,0,0,0)$, which I show using the same encoding as that for the previous sequence.  
\end{remark}

Observe how our modification of the spelling function allowed us to enlarge its domain from $\mathbb{N}$ to $\mathbb{S}$ without requiring intermediate modifications to go from $\mathbb{N}$ to $\mathbb{Z}$ to  $\mathbb{Q}$.  This is because the recursive nature of the spelling function implies that if it can spell negative numbers, it can also spell negative exponents.  Indeed, enlargement of the domain of the spelling function from $\mathbb{N}$ to merely $\mathbb{Z}$, or even from $\mathbb{N}$ to merely $\mathbb{Q}$, would require artificial restrictions imposed to limit the recursivity of the function.  In other words, with recursive prime factorizations it is easier and more straightforward to go directly from representing natural numbers to representing superrational numbers, than to go through intermediate sets to represent the superrationals.  When we modify the spelling function so we can spell negative numbers, we get the ability to spell all superrational numbers ``for free.''

Before we define the standard minimal RPF superrational interpretation, let us give a name to its underlying language.

\begin{definition} \label{def:LDQmin}
The \emph{standard quasiminimal RPF language}, denoted by $\mathcal{D}_{r_{\text{qmin}}}$, is the codomain of $\upgamma_{\mathbb{S}_{r}}$.
\end{definition}

At long last we arrive at the definition of $\text{RPF}_{\mathbb{S}_{r_{\text{min}}}}$.

\begin{definition}\label{def:StdMinimalRPFSuperrationalInterpretation}
   The \emph{standard minimal RPF superrational interpretation}, denoted by  $\text{RPF}_{\mathbb{S}_{r_{\text{min}}}}$, is the interpretation $(\mathcal{D}_{r_{\text{qmin}}}, \mathbb{S}, \upgamma^{-1}_{\mathbb{S}_{r}})$.
\end{definition}

Throughout the rest of this paper, a certain Dyck word occurs so often that I give it a special name.
\begin{definition} \label{def:EmptyPair}
An \emph{empty pair} is the string ${()}$.
\end{definition}

I now demonstrate that $\DrQmin \subsetneq \mathcal{D}$.
\begin{theorem} \label{thm:DrQmin_is_proper_subset_of_D}
The standard quasiminimal RPF language \DrQmin is a proper subset of the Dyck language $\mathcal{D}$.
\end{theorem}
\begin{proof}
I will show that $\DrQmin \subseteq \mathcal{D}$ and then provide an element of $\mathcal{D} - \DrQmin$.

Comparing the definitions of the nonsurjective precursors of the natural and superrational spelling functions, we see that they spell natural numbers identically.  From Theorem~\ref{thm:StdMinNaturalRPFIsProperSubsetOfD}, we already know that the codomain \DrMin of \gNr is a proper subset of $D$; thus the set of natural numbers in \DrQmin  is a proper subset of $\mathcal{D}$.
Now let set $T = \mathcal{D} - \DrMin$. I first consider the members of $T$ which are spellings of negative integers.  As we can see from the definition of the precursor of \gSr, these are spelled by appending exactly one empty pair to the spellings of their absolute values.  The spellings of all negative integers are therefore in $\mathcal{D}$, since the concatenation of two Dyck words is also a Dyck word .   The only remaining members of $T$ to consider are spellings of numbers $s \in \mathbb{S} - \mathbb{Z}$, where $\mathbb{Z}$ is the set of integers.  But since these spellings arise recursively from appending one empty pair to the spelling of a positive number in $\mathbb{S}$, we see by Lemma~\ref{lem:ChunkContainsDyckWordLenTwoLess} that words so produced must be in $\mathcal{D}$.

An example of a word in $\mathcal{D}$ but not in \DrQmin is  $w = ()()()$.
$w$ is clearly a Dyck word from Definition~\ref{def:DyckLanguage}.  But  $w$ cannot be in the codomain of \gSr, because additional empty pairs only occur through the recursive spellings of constituent chunks representing negative numbers.  Thus no two consecutive empty pairs will be appended to the end of a chunk.  The presence of $()()()$ in \DrQmin would imply an empty pair was appended to the spelling of a positive number, which is a contradiction.    This fact can be illustrated visually by considering the following two proposed recursion trees for the spelling of $-1$; the two nodes whose presence implies an invalid spelling are enclosed in boxes.

\begin{center}
\begin{tikzpicture}[level distance=1.0cm, sibling distance=0.4cm]
  \node (leftroot) at (1,.53) {
    \Tree [.$\upgamma_{\mathbb{S}_{-1}}$ 
             [.$\upgamma_{\mathbb{S}_{1}}$ ] 
             $()$ 
         ]
  };
  
  \node (rightroot) at (5.6,0) {
    \Tree [.$\upgamma_{\mathbb{S}_{-1}}$
             [.$\upgamma_{\mathbb{S}_{1}}$
                \edge; {\fbox{$\upgamma_{\mathbb{S}_{1}}$}}
                \edge; {\fbox{$()$}}
             ]
             $()$
         ]
  };
\end{tikzpicture}
\end{center}

In both proposed spelling trees, empty pairs arising from negation are displayed as literals in order to distinguish them. 
The tree on the left corresponds to the correct spelling of $-1$ as $()()$, while the tree on the right corresponds to an  incorrect spelling of $-1$ as $()()()$.  Specifically, the nodes enclosed in boxes would imply that the spelling of $1$ entails itself, yielding an infinite recursion chain; in reality, the spelling of $1$ \emph{terminates} a recursion chain.

\end{proof}

\bigskip

We can now be certain that $\gSr$ has an inverse such that the domain \DrQmin of that inverse is a proper subset of $\mathcal{D}$.

\subsection{The RPF Superrational Evaluation Functions and Their Relationship to Dyck Inflations}

\begin{definition}\label{def:StdDyckCompleteRationalEvaluation}
Let $w \in \mathcal{D}$.  Then the \emph{standard RPF superrational evaluation} of $w$ is the function $\upalpha_{\mathbb{S}_{r}} : \mathcal{D} \rightarrow \mathbb{S}$, where
\begin{itemize}
\item
If $w = \upepsilon$, $\upalpha_{\mathbb{S}_{r}}(w) = 0$.
\item
Otherwise, let $w'$ and $z$ be Dyck words where $w = w' \myfrown z$, with $z$ being the longest suffix of $w$ such that $w' \neq \upepsilon$ and $z$ contains only empty pairs.  Then
\begin{equation}
\upalpha_{\mathbb{S}_{r}}(w) = (-1)^{\text{dim}(z)}  \prod_{i = 1}^{\text{dim}(w')} p_{i}^{\upalpha_{\mathbb{S}_{r}}(d_{i})},\nonumber
\end{equation}
where $d_{i}$ is the content of the $i$th chunk in $w'$.
\end{itemize}
\end{definition}
\begin{example}
Let us evaluate the Dyck word $w = {(())((())())()}$ as a standard RPF superrational number.
Note that $w = w{'} \myfrown z$, where $w{'} = {(())((())())}$ and $z = {()}$.  Thus we have
\begin{equation}
\upalpha_{\mathbb{S}_{r}}(w) = (-1)^{\text{dim}(z)}  \prod_{i = 1}^{\text{dim}(w')} p_{i}^{\upalpha_{\mathbb{S}_{r}}(d_{i})},\nonumber
\end{equation}
with $d = {({'()'}, {'(())()'})}$.  Therefore
\begin{equation} \nonumber
\begin{aligned}
& \upalpha_{\mathbb{S}_{r}}(w) =
{(-1)}^{1}     p_{1} ^{   \upalpha_{\mathbb{S}_{r}}({'()'})}
p_{2}^{      \upalpha_{\mathbb{S}_{r}}({'(())()'})                  }                   \\
& =   -p_{1}^{   p_{1}^      {         \upalpha_{\mathbb{S}_{r}}(\upepsilon)          }       }     
p_{2}^{    {(-1)}^{1}        p_{1}^{          \upalpha_{\mathbb{S}_{r}}({'()'})          }                             }                 \\
& =   - p_{1}^{   p_{1}^{0}        }  p_{2}^{   - p_{1}^{     p_{1}^{    \upalpha_{\mathbb{S}_{r}}(\upepsilon)   }             }      }   
=   - p_{1} p_{2}^ {- p_{1}^{   p_{1}^{0}    }} 
= - p_{1}p_{2}^{-2}             \\
& =  -2 \cdot 3^{-2}  = -{\frac{\;2\;}{\;9\;}} = -0.2222\ldots\;.
\end{aligned}
\end{equation}
\end{example}

\begin{definition} \label{def:StdMinimalRationalEvaluation}
The \emph{standard minimal RPF superrational evaluation function}, denoted by $\upalpha_{\mathbb{S}_{r_{\text{min}}}}$, is the restriction of $\upalpha_{\mathbb{S}_{r}}$ to $\mathcal{D}_{r_{\text{qmin}}}$.
\end{definition}

We know that domain of $\aSrMin$ and the codomain of $\gSr$ are the same set, and soon I will go further to identify $ \gSr$ and $\aSr$ as mutual inverses.    But I will make my task easier by first discussing the concept of  \emph{inflations}.

\begin{remark}
Inflations are also relevant to research topics in recursive prime factorizations beyond the scope of this paper.  For example, readers may wish to explore using the Dyck-complete natural interpretation to establish a total ordering of the Dyck language by mapping words in $\mathcal{D}$ uniquely (though nonbijectively) to the set of natural numbers.  Dyck words could first be ordered numerically, then according to a function of nesting levels of chunks in words, and finally by semilength.  Alternative orderings making use of Dyck inflations are conceivable.
\end{remark}

Definition~\ref{def:MinimalInterpretation} tells us what we are to understand by the term ``minimal representation,'' and of course a nonminimal representation is a representation that is not minimal.  But let us view the concept of nonminimality from a different perspective.  In the decimal system, every member of the infinite sequence of strings $(102, 0102, 00102, \ldots)$ is understood to represent the number one hundred and two.  We have no problem, for instance, determining the quantity represented by 000102 on a car odometer. Nevertheless there is one particular string, 102, that is the minimal word representing one hundred and two; all others may be regarded as a result of ``inflating'' the minimal word with zeroes without changing the quantity it represents.  The situation is similar for words in nonminimal RPF languages, except that the inflations involve empty pairs rather than zeros and there almost always exist multiple locations in the word where inflation may occur.

Consider the three Dyck words $w = ()()(())$, $x = ()()(())()$ and $y = ()()(()())$.  Applying the standard natural evaluation function to each of these, we find that
\begin{equation} \nonumber
\upalpha_{\mathbb{N}_{r}}(w) = \upalpha_{\mathbb{N}_{r}}(x)  = \upalpha_{\mathbb{N}_{r}}(y) = 5.
\end{equation}

If we apply the standard natural spelling function to each of the above evaluations, we discover that
\begin{equation} \nonumber
\upgamma_{\mathbb{N}_{r}}({\upalpha_{\mathbb{N}_{r}}(w)}) = 
\upgamma_{\mathbb{N}_{r}}({\upalpha_{\mathbb{N}_{r}}(x)})  =
\upgamma_{\mathbb{N}_{r}}({\upalpha_{\mathbb{N}_{r}}(y)})
 = {'()()(())'} = w,
\end{equation}
so that $w$ is the only Dyck word of the three where the word equals the spelling of its evaluation.  Also notice that of $w$, $x$ and $y$, $w$ is the word of shortest length;
indeed, $()()(())$, being a member of $\mathcal{D}_{r_{\text{min}}}$, is the shortest Dyck word evaluating to 5 in the standard natural interpretation.

\begin{definition} \label{def:Inflation}
Let $d$ and $d'$ be Dyck words. $d$ is said to be an \emph{inflation} of $d'$ if $\upalpha_{\mathbb{N}_{r}}(d) = \upalpha_{\mathbb{N}_{r}}(d')$ and {${\mid} d {\mid} > {\mid} d' {\mid}$}, where ${\mid} s {\mid}$ denotes the length of string $s$.  If we do not wish to specify {$d'$}, we may simply say \emph{d is an inflation}, implying that some $w \in \mathcal{D}$ exists such that $d$ is an inflation of $w$.
\end{definition}
\begin{remark}
Dyck word $w$ is an inflation if and only if $w \neq \upgamma_{\mathbb{N}_r}(  \upalpha_{\mathbb{N}_{r}}(w))$.
\end{remark}
\begin{definition} \label{def:deflation}
Let $d'$ and $d$ be Dyck words such that $d$ is an inflation of $d'$.  Then we say $d'$ is a \emph{deflation} of $d$.
\end{definition}

Words in $\mathcal{D}_{r_{\text{min}}}$ cannot be inflations, as the only empty pairs they contain are those essential as separators to collectively designate indices of prime numbers contributing to the prime factorization (or, in the case of (), to spell the number 1).
Consider an arbitrary Dyck word $w$.  We can form a sequence with $w$ as its first term, with every subsequent term being a deflation of its predecessor; the longest possible such sequence must have a member of $\mathcal{D}_{r_{\text{min}}}$ as its last term.

\begin{definition} \label{def:Collapse}
Let $w$ and $w'$ be distinct Dyck words. We say \emph{w collapses to ${w'}$} if  $\upalpha_{\mathbb{N}_{r}}(w) = \upalpha_{\mathbb{N}_{r}}(w')$ and $w' \in \mathcal{D}_{r_{\text{min}}}$.
\end{definition}
\begin{example}
Dyck word $w = {(()()(())()())()()()}$ collapses to ${w{'} = (()()(()))}$; the standard natural evaluation of both $w$ and $w{'}$ is 32, and $w{'}$ is a member of $\mathcal{D}_{r_{\text{min}}}$.
\end{example}

The following gives us terminology to talk about strings consisting solely of empty pairs.
\begin{definition} \label{def:ContiguousEmptyPairs}
Let $w = \bigoplus_{i=1}^{k}{'()'}$ for some $k \in \mathbb{N}$.  Then we say  $w$ is \emph{the string of k contiguous empty pairs}.  If $k = 0$, we say the string is \emph{trivial}.
\end{definition}
Thus the trivial string of zero contiguous empty pairs is $\upepsilon$,  the string of one contiguous empty pair is ${()}$, the string of two contiguous empty pairs is ${()()}$, and so on.

The following definition allows us to quantify ``how inflated'' we consider a given Dyck word to be.
\begin{definition} \label{def:InflationaryDegree}
Let $w$ be a  Dyck word. The \emph{standard inflationary degree} of $w$, denoted by $\text{dinf}_{r}(w)$, is the largest integer $n$ such that a string of $n$ contiguous empty pairs can be deleted from $w$ to yield a Dyck word $w'$ satisfying the equation
\begin{equation} \nonumber
\upalpha_{\mathbb{N}_{r}}(w') = \upalpha_{\mathbb{N}_{r}}(w).
\end{equation}
\end{definition}
\begin{example}
Let $w = {(()()()())()}$. The longest substring that can be deleted from $w$ to yield a Dyck word $w'$ with the same standard natural evaluation as that of $w$ is ${()()()}$, the string of 3 contiguous empty pairs.  Thus  $\text{dinf}_{r}({w}) = 3$.
\end{example}

It is useful to have a way of depicting where inflations can occur in standard recursive prime factorizations.
\begin{definition}  \label{def:inflatability_points}
Let $w$ be a nontrivial Dyck word, and let $e$ be an empty pair. The \emph{inflatability points} of $w$ are locations where appending $e$ yields an inflation $w'$ of $w$, with the restriction that prepending $e$ to another empty pair is prohibited.  We call an inflatability point located inside a word \emph{internal}, whereas the inflatability point located immediately after the last symbol of the word is called \emph{external} (or \emph{terminal}).
\end{definition}

\begin{example}
Consider the Dyck natural word $()$, which evaluates to 1.  Depicting inflatability points according to their order of occurrence from left to right, the notation  $()1$ asserts that $()$ has only a single inflatability point, the terminal one.  This is equivalent to asserting that the only place where the Dyck natural representation of 1 can be inflated is at its very end.
\end{example}

\begin{example}
Consider the Dyck natural word $(()(()))$, which evaluates to  8.  $(()(()1)2)3$ asserts that $(()(()))$ has three inflatability points, two of which are internal and one which is terminal.  This is equivalent to asserting that the only places where the Dyck natural representation of 8 can be inflated are at the locations shown.
\end{example}

\begin{example}
Consider the \emph{nonminimal} Dyck natural word $(()(()))()()()$, which evaluates to  8.  $(()(()1)2)()()()3$ is equivalent to asserting that $(()(()))()()()$ has three inflatability points, two of which are internal and one which is terminal.  Although empty pairs could be inserted at other places, namely between any two inflationary empty pairs, we are prohibited from doing so.
\end{example}

These considerations lead to the following theorem.
\begin{theorem}  \label{thm:SufficiencyOfInfPoints}
The set of inflatability points is sufficient to show every place where inflation may take place in a nontrivial Dyck word $w$, regardless of whether $w$ is itself an inflation.
\end{theorem}
\begin{proof}
To see this, we need only recall the alternative definition of \DrMin as that subset of Dyck words $w$ both avoiding $')())'$ and not ending with $')()'$.  Internal inflations correspond to the former constraint, while terminal inflations correspond to the latter.  If we insert $k \ge 1$ contiguous empty pairs between two consecutive right parentheses in a nontrivial Dyck word, at its end, or both, the result will be an inflation of its minimal counterpart, therefore not altering its evaluation as a natural  number;  in that regard, it does not matter whether $k$ is 1 or 1000.
\end{proof}

And so we obtain an alternative definition of the standard minimal RPF superrational interpretation.
\begin{definition} \label{def:StandardQuasiminimalRationalLanguage}
The \emph{standard quasiminimal RPF language}, also called \emph{the set of Dyck superrational numbers} and denoted by $\mathcal{D}_{r_{\text{qmin}}}$, is the set of all Dyck words $d$ such that $\text{dinf}_{r}(d) \leq 1$.
\end{definition}
\begin{remark}
The designation \emph{quasiminimal} and subscript  $\text{qmin}$ are due to the fact that \DrQmin represents $\mathbb{N}$ nonminimally, whereas it represents $\mathbb{S}$ minimally. In hindsight, using $\mathcal{D}_{r_{\mathbb{N}}}$ and  $\mathcal{D}_{r_{\mathbb{S}}}$ to respectively denote what I call \DrMin and \DrQmin  in this paper would have been more intuitive.  I could even further simplify the notation by omitting the permutation subscript $r$, yielding $\mathcal{D}_{\mathbb{N}}$ and $\mathcal{D}_{\mathbb{S}}$. I may indeed switch to that notation in future revisions.
\end{remark}

The only nodes to which a single inflationary empty pair cannot be appended in the course of spelling a number are those corresponding to $\upepsilon$.

The standard quasiminimal RPF language $\mathcal{D}_{r_{qmin}}$ has an alternative, non-numerical definition equivalent to the two definitions already given. I leave the following as a proposition, as I do not need the definition to further develop the concepts and structures in this paper.  For those readers wishing to demonstrate the equivalence of the purely grammatical definition with the other two, Definition~\ref{def:inflatability_points} and Theorem~\ref{thm:SufficiencyOfInfPoints} may be helpful.

\begin{prop} \label{prp:AltStandardQuasiminimalRationalLanguage}
$\mathcal{D}_{r_{\text{qmin}}}$ is the set of all $d \in \mathcal{D}$ such that $d$ fulfills the following two criteria:
   \begin{itemize}
      \item $d$ avoids  ${)()())}$.
      \item $d$ does not end with ${)()()}$.
   \end{itemize}
\end{prop}
\begin{remark}
Essentially, the proposition captures the fact that appending an optional single terminal empty pair is sufficient to permit the inclusion of negative coefficients, which in this recursive system implies the extension of $\mathbb{N}$ to $\mathbb{S}$.  Using the nonnumerical alternative definition of $\mathcal{D}_{r_{\text{qmin}}}$, the definitions of ``inflation'' (Def.~\ref{def:Inflation}), ``collapse'' (Def.~\ref{def:Collapse}) and ``inflationary degree'' (Def.~\ref{def:InflationaryDegree}) could be restated so that they do not refer to evaluation function $\upalpha_{\mathbb{N}_{\text{r}}}$, thus eliminating the necessity to perform prime factorization in the course of their application.
\end{remark}

I have now developed sufficient conceptual scaffolding upon which to complete my argument that $\gSr$ and $\aSrMin$ are mutual inverses.  I shall now write from a higher-level point of view, favoring conversational language over a further onslaught of lemmas and theorems, with the intent to persuade rather than inundate.

I have already established that the inverse \gSr[][{-1}] of \gSr exists.  Since \gSr has an inverse, I can prove that \aSrMin  and \gSr are mutual inverses by showing that for all $s \in \mathbb{S}$
\begin{equation} \nonumber
\aSrMin[\gSr[s]] = s.
\end{equation}

From Theorem~\ref{thm:StdNatEvalAndSpellingAreInverses}, we know that $\aNrMin[\gNr[n]] = n$ for all $n \in \mathbb{N}$.  If we compare the nonsurjective precursor of the natural spelling function $\gNr[][\prime]$ with its superrational counterpart  $\gSr[][\prime]$, we see that the spelling of a natural number by $\gSr$ is identical to the spelling of the same number by $\gNr$.  Furthermore, if we compare the noninjective precursor of $\aNrMin$ with its superrational counterpart, we see that these functions evaluate Dyck words $d$ of inflationary degree 0  such that they yield the same evaluation, with the machinery of the superrational version simplifying to that of its natural counterpart.  Thus
\begin{equation} \nonumber
\forall n \in \mathbb{N} \quad   [\aSrMin[n] = \aNrMin[n]].
\end{equation}

It remains to be shown that   $\aSrMin[\gSr[s]] = s$ for all $s \in \mathbb{S} \setminus \mathbb{N}$.

For the case of $s = -1$, manual computation gives
\begin{equation} \nonumber
\aSrMin[\gSr[s]]  =     \aSrMin[{'()()'}]      = -1 = s.
\end{equation}

The only superrational numbers left to be considered are those expressible as products of finite power towers of primes with negatable exponents.

Looking at the definition of the nonsurjective precursor $\gSr[][\prime]$ of $\gSr$, we see $\gSr$ spells superrational numbers in a canonical way.  Numbers are treated as products of finite power towers of primes with negatable exponents, outputting a sequence of chunks.  Successive prime bases are encoded as chunk ordinality, sign is encoded as chunk cardinality (where the presence or absence of negation results in the respective presence or absence of a corresponding single inflationary empty pair), and prime-base exponentiation is encoded by nesting of chunks.  Furthermore, because $\gSr$ is surjective, it never produces Dyck words containing more empty pairs than those required to represent the number.

Turning to the definition of the noninjective precursor $\aSr$ of $\aSrMin$, we see $\aSrMin$ also operates in a canonical manner, evaluating members of \DrQmin  by ``reading off'' its input such that sequences of chunks  are translated into equivalent finite power towers of primes with negatable exponents.  The order of chunks is treated as an encoding of prime bases, the number of chunks is treated as an encoding of sign (with sign being determined by the presence or absence of a single inflationary empty pair), and the nesting of chunks is treated as an encoding of prime-base exponentiation.  Furthermore, because the domain of $\aSrMin$ is restricted to $\DrQmin$, only minimal words and first-degree inflations are permitted; thus the function will never produce evaluations of Dyck words in which multiple negations are applied to a number at a particular location within the product of power towers.

There is one aspect in which a comparison of $\gSr$  with the evaluation function $\aSrMin$ may give the appearance that the two are not mutual inverses: recursion chains for $\gSr$ terminate with the spelling of 1 as $()$, whereas recursion chains for $\aSrMin$ terminate by evaluating $\upepsilon$ as 0.  However, this asymmetry only reflects a choice of different base cases for the two functions; $\upepsilon$ is a trivial substring of $()$, so that  \gSr and \aSrMin still undo each other.

And so we do not have to resort to further formal definitions, lemmas and proofs to show that $\aSrMin$ and $\gSr$ are mutual inverses; examination of those functions' machinery reveals that the relationship arises naturally.  Recursion places its canonical stamp on both superrational numbers and the Dyck words minimally representing them, enforcing the correspondence between the two sets.

\subsection{Why a Second Incremental Inflation Set Is Not Particularly Useful}

Having modified $\text{RPF}_{\mathbb{N}_{r_\text{min}}}$ to yield $\text{RPF}_{\mathbb{S}_{r_\text{min}}}$ by permitting first-degree inflations, we might imagine that further modifying it to permit second-degree inflations would buy us additional representational power, so that we could use terminal empty parenthesis pairs to encode additional two-state attributes---for example, an attribute called ``spin,'' which could either be clockwise or counterclockwise.  Furthermore, we might imagine that the attributes of sign and of spin would be \emph{orthogonal}, i.e., that the value of the sign could be determined without having to know that of the spin (and vice versa). But such a happy state of affairs is not the case, as the following theorem states for standard RPF interpretations.

\begin{theorem} \label{thm:QuasiQuasiMinimalNotUseful}
Let $R_{0}$ be the standard minimal RPF language of 0-degree inflations, let $R_{1}$ be the standard quasiminimal RPF language of inflations of degree $\le 1$, and let $R_{2}$ be the standard RPF language of inflations of degree $\le 2$.  Then $R_{2}$ is no more powerful than $R_{1}$ for encoding orthogonal binary attributes.
\end{theorem}
\begin{proof}
$R_{1}$ includes minimal RPF words and their single terminal inflations, so it can be successfully used for encoding a single binary attribute modifying the evaluation of the word.  For example, suppose the attribute is $(\text{north}, \text{south})$, such that $(())$ means ``$2$ units due north,'' whereas its terminal inflation $(())()$ evaluates to ``$2$ units due south.''  This is indeed possible using a subset of the Dyck language confined to inflations of degree $\le 1$. But now suppose we try to extend the concept by using $R_{2}$ so we can encode an additional orthogonal binary attribute, say, that of $(\text{west}, \text{east})$,with the next-to-last empty parenthesis pair signifying north versus south and the last empty parenthesis pair signifying west versus east.  We might for instance na\"ively claim that whereas $(())$ represents ``$2$ units northwest'', $(())()()$ represents ``$2$ units southeast.''  Under such an interpretation, however, the word $(())()$ is ambiguous; without further information aside from the word itself and its interpretation, we cannot tell whether it signifies 2 units northeast or 2 units southwest.  In other words, we can only be sure of the values of words that are minimal or have ${'}()(){'}$ as a proper suffix .  We might be content to impose a convention for disambiguation, so that all occurrences of a minimal word concatenated with a single empty parenthesis pair would be resolved in favor of, say, south rather than east.  But to impose such a convention would destroy the orthogonality of the two attributes.
\end{proof}
\begin{remark}
The theorem states that the set of inflations of degree $\le 2$ is no more useful for encoding orthogonal binary attributes than the set of inflations of degree $\le 1$, but says nothing about sets including inflations of even higher degree.  Obviously, though, permitting even longer strings of inflationary empty pairs does not eliminate the ambiguity.
\end{remark}

\section{Dyck-complete Interpretations} \label{loc:DyckCompleteInterpretations}
While a minimal RPF natural interpretation and its corresponding minimal RPF superrational interpretation are sufficient to represent all natural numbers and all superrational numbers respectively,  and while each of these interpretations enjoys the property that a bijection exists between its underlying language and its target set,
we can generalize our notion of RPF languages to permit all possible inflations.  The result is the Dyck language, which underlies all generalized RPF interpretations, regardless of whether they are natural or superrational, or which prime permutations determine the ordering of factors in their words.

\begin{definition}\label{def:DyckComplete}
An interpretation is \emph{Dyck-complete} if its underlying language is $\mathcal{D}$.
\end{definition}

We already have our evaluation functions for the standard Dyck-complete natural and superrational interpretations; these are $\upalpha_{\mathbb{N}_{r}}$ and $\upalpha_{\mathbb{S}_{r}}$, respectively.

\begin{definition} \label{def:StdDyckCompleteNaturalInterpretation}
The \emph{standard RPF Dyck-complete natural interpretation}, denoted by $\text{RPF}_{\mathbb{N}_{r}}$,  is the interpretation $(\mathcal{D}, \mathbb{N},  \upalpha_{\mathbb{N}_{r}})$.
\end{definition}

\begin{definition} \label{def:StdDyckCompleteRationalInterpretation}
The \emph{standard RPF Dyck-complete superrational interpretation}, denoted by $\text{RPF}_{\mathbb{S}_{r}}$,  is the interpretation $(\mathcal{D}, \mathbb{S},  \upalpha_{\mathbb{S}_{r}})$.
\end{definition}

These interpretations are nonminimal, their evaluation functions being noninjective.  There is a spelling function associated with each interpretation---$\upgamma_{\mathbb{N}_{r}}$  for $\text{RPF}_{\mathbb{N}_{r}}$, $\upgamma_{\mathbb{S}_{r}}$  for $\text{RPF}_{\mathbb{S}_{r}}$---but words do not generally equal the spellings of their evaluations.  Indeed, the evaluation functions define equivalence relations partitioning the underlying languages of their interpretations into equivalence classes, where each equivalence class contains Dyck words that evaluate to the same number.  For example, let $R_{\mathbb{N}}$ be the equivalence relation
\begin{equation} \nonumber
R_{\mathbb{N}} = \{(w_{1}, w_{2}) \in \mathcal{D}^{2} \mid \upalpha_{\mathbb{N}_{r}}(w_{1}) = \upalpha_{\mathbb{N}_{r}}(w_{2}) \}.
\end{equation}
Then we can identify each equivalence class as $T_{n}$, where every member of the class evaluates to the natural number $n$.  Thus $T_{0} = \{\upepsilon\}$, $T_{1} = \{'()', '()()', '()()()', \ldots\}$, and so on.  In each equivalence class, there is exactly one word that is in the codomain of the spelling function; it is thus the only member of its class such that it is equal to the spelling of its evaluation.

The same considerations apply for the Dyck-complete superrationals, but the inclusion of negation in the recursive system complicates the partitioning; successive inflations at any given inflation point result in a ``toggling'' or alternation of corresponding evaluations.  For example, $((()))$ evaluates to $4$, $((()()))$ evaluates to $\sqrt{2}$,  $((()()()))$ evaluates to $4$, $((()()()()))$ evaluates to $\sqrt{2}$, and so on.

Figure~\vref{fig:EulerRPF}
shows the hierarchy of the languages $L$ underlying standard RPF interpretations, with the interpretations designated by subscripts.   Also shown are the same sets as named according to their status as subsets of the Dyck language: the standard Dyck minimals, the standard Dyck quasiminimals, and the Dyck language itself.

\begin{figure}[!htb]  
\centering
{
\Large
\tikzset{
node distance=3cm, 
every state/.style={thick, fill=gray!5}, 
initial text=$ $, 
}
\begin{tikzpicture}
\draw [thick, fill=gray!0] (0,-1) arc (-90:270:5cm and 3cm);
\draw [thick, fill=gray!8] (0,-0.5) arc (-90:270:3.5cm and 2cm);
\draw [thick, fill=gray!20] (0,0) arc (-90:270:3cm and 1cm);
\node [yshift=1cm] (0,0) {$L_{\text{RPF}_{\mathbb{N}_{r_{\text{min}}}}} = \mathcal{D}_{r_{\text{min}}}$};
\node [yshift=2.5cm] (0,0) {$L_{\text{RPF}_{\mathbb{S}_{r_{\text{min}}}}} = \mathcal{D}_{r_{\text{qmin}}}$};
\node [yshift=4cm] (0,0) {$L_{\text{RPF}_{\mathbb{N}_{r}}} = 
L_{\text{RPF}_{\mathbb{S}_r}} =
\mathcal{D}
$};
\end{tikzpicture}
}
\caption{Euler diagram illustrating the hierarchy of standard RPF languages.
}
\label{fig:EulerRPF}
\end{figure}
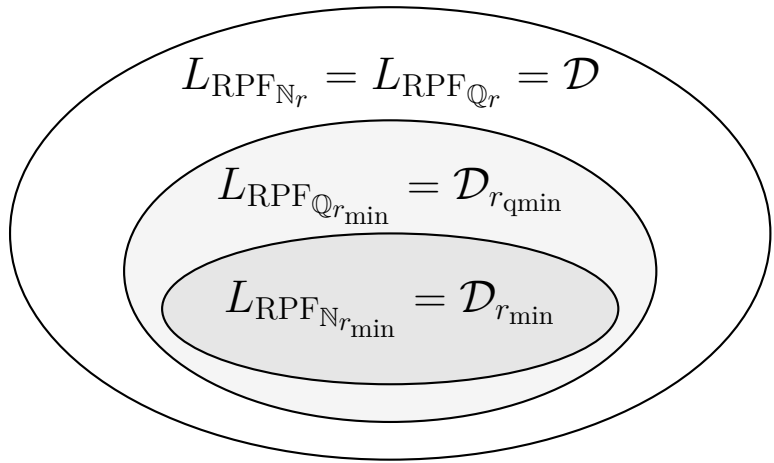

\section{Conclusion: Research Directions and Possible Applications}

I conclude with suggestions for further study of recursive prime factorizations, as well as possible applications in mathematics and computer science.

\subsection{Investigate the Properties of Stripes} \label{loc:InvestigateStripes}
\begin{definition} \label{def:Stripe}
Let $k$ be a natural number.  The \emph{stripe of semilength k}, denoted by $\uptheta(k)$, is the finite subsequence $s$ of all $n$ in $(0,1,2,3, \ldots)$ such that $n$ satisfies the equation ${\mid}\upgamma_{\mathbb{N}_{r}}(n){\mid} = 2k$, where ${\mid}w{\mid}$ is the length of string $w$.  We call $\uptheta(0)$ and $\uptheta(1)$ \emph{trivial stripes}.
\end{definition}
Let $S$ be the sequence of the first four nontrivial stripes.  Then $S$ is
\begin{equation}\nonumber
((2), (3,4), (5,6,8,9,16), (7,10,12,15,18,25,27,32,64,81,256,512,65536)).
\end{equation}
We see interesting patterns in $S$.
For example, the first member of the $i$th stripe in $S$ is the prime number $p_{i}$, and no term exists in $S_{i}$ such that its spelling entails the spelling of an earlier term.
Also,  the sequence of last terms of stripes in $S$ may be expressed as $(2, 2^{2}, 2^{2^{2}},2^{2^{2^{2}}})$, i.e., as base-2 tetration. 

Though interesting, none of the aforementioned patterns are surprising once we consider them from the perspective of recursive prime factorizations.  Base-2 tetration in nontrivial stripes arises from working with a binary symbol set in which nesting corresponds to prime-base exponentiation. No spelling of terms within stripes entails the spelling of earlier terms, as two words in a recursion chain cannot be of the same semilength.  Prime numbers are the first terms in nontrivial stripes because their representations are words in which all but one chunk are empty and therefore evaluate to 1, with the nonempty chunk only having a single level of nesting; all other terms within their stripes must therefore be greater in numerical value.

We may regard the infinite sequence of stripes of successive semilengths as a table $T$. Because the finite sequences that comprise the table have different lengths, $T$ has jagged edges.  The table displays additional patterns involving integer factorization, particularly as the row numbers of the table increase.  For example, let us denote the number residing at the $i$th row and $j$th column by $T_{i,j}$, with subscripts starting at 1.  Fixing the value of $j$, entries in higher rows appear to follow the pattern  $jp_{i-3}$.  For example, the number located at row 18, column 2 is $94 = 2p_{15}$, while the number located at row 19, column 3 is $159 = 3p_{16}$.

What other patterns characterize stripes? 
\begin{remark}
Ancillary file dyck\_stripes.csv contains a table of Dyck-natural stripes, populated by enumerating all natural numbers up to and including $16383$.
\end{remark}

\subsection{Explore Learnability Implications of Dyck Superrationals Corresponding to Lossless Factorization Trees}

Pipelines culminating in recursive factorization trees are studied in \cite{iudelevich2022tree,contucci2025statistical,breccia2025testing}, being illustrated in the third by

\begin{equation*}n=2^{64}\cdot 3=2^{2^{2\cdot 3}}\cdot 3 \mapsto  \tau(n)=\{2:\{2:\{\{2:1\},\{3:1\}\}\},3:1\}\mapsto  \begin{tikzpicture}[baseline={([yshift=-.5ex]current bounding box.center)}]
\node[inner sep=1pt,circle,draw,thick,fill=gray] (0) at (0,0) {\small $n$};
\node[inner sep=1pt,circle,draw,thick] (1) at (-0.5,0.5) {\small $2$};
\node[inner sep=1pt,circle,draw,thick] (2) at (0.5,0.5) {\small $3$};
\node[inner sep=1pt,circle,draw,thick] (12) at (-0.5,1) {\small $2$};
\node[inner sep=1pt,circle,draw,thick] (121) at (-1,1.5) {\small $2$};
\node[inner sep=1pt,circle,draw,thick] (122) at (0,1.5) {\small $3$};
\draw[thick] (122) -- (12) -- (1) -- (0) -- (2);
\draw[thick] (121) -- (12);
\end{tikzpicture}\label{eq:fig}\end{equation*}

References \cite{contucci2025statistical} and\cite{breccia2025testing} go further to encode the trees as Dyck words; in the discussion which follows, let us represent such trees in accordance with the convention shown above, the first fourteen trees therefore being 

\[\begin{tikzpicture}
\node[inner sep=1pt,circle,draw,thick,fill=gray] (0) at (0,0) {\small $1$};
\end{tikzpicture}\quad\begin{tikzpicture}
\node[inner sep=1pt,circle,draw,thick,fill=gray] (0) at (0,0) {\small $2$};
\node[inner sep=1pt,circle,draw,thick] (1) at (0,0.5) {\small $2$};
\draw[thick] (0) -- (1);
\end{tikzpicture}\quad\begin{tikzpicture}
\node[inner sep=1pt,circle,draw,thick,fill=gray] (0) at (0,0) {\small $3$};
\node[inner sep=1pt,circle,draw,thick] (1) at (0,0.5) {\small $3$};
\draw[thick] (0) -- (1);
\end{tikzpicture}\quad\begin{tikzpicture}
\node[inner sep=1pt,circle,draw,thick,fill=gray] (0) at (0,0) {\small $4$};
\node[inner sep=1pt,circle,draw,thick] (1) at (0,0.5) {\small $2$};
\node[inner sep=1pt,circle,draw,thick] (2) at (0,1) {\small $2$};
\draw[thick] (0) -- (1) -- (2);
\end{tikzpicture}\quad\begin{tikzpicture}
\node[inner sep=1pt,circle,draw,thick,fill=gray] (0) at (0,0) {\small $5$};
\node[inner sep=1pt,circle,draw,thick] (1) at (0,0.5) {\small $5$};
\draw[thick] (0) -- (1);
\end{tikzpicture}\quad\begin{tikzpicture}
\node[inner sep=1pt,circle,draw,thick,fill=gray] (0) at (0,0) {\small $6$};
\node[inner sep=1pt,circle,draw,thick] (1) at (-0.25,0.5) {\small $2$};
\node[inner sep=1pt,circle,draw,thick] (2) at (0.25,0.5) {\small $3$};
\draw[thick] (2) -- (0) -- (1);
\end{tikzpicture}\quad\begin{tikzpicture}
\node[inner sep=1pt,circle,draw,thick,fill=gray] (0) at (0,0) {\small $7$};
\node[inner sep=1pt,circle,draw,thick] (1) at (0,0.5) {\small $7$};
\draw[thick] (0) -- (1);
\end{tikzpicture}\quad\begin{tikzpicture}
\node[inner sep=1pt,circle,draw,thick,fill=gray] (0) at (0,0) {\small $8$};
\node[inner sep=1pt,circle,draw,thick] (1) at (0,0.5) {\small $2$};
\node[inner sep=1pt,circle,draw,thick] (2) at (0,1) {\small $3$};
\draw[thick] (0) -- (1) -- (2);
\end{tikzpicture}\quad\begin{tikzpicture}
\node[inner sep=1pt,circle,draw,thick,fill=gray] (0) at (0,0) {\small $9$};
\node[inner sep=1pt,circle,draw,thick] (1) at (0,0.5) {\small $3$};
\node[inner sep=1pt,circle,draw,thick] (2) at (0,1) {\small $2$};
\draw[thick] (0) -- (1) -- (2);
\end{tikzpicture}\quad\begin{tikzpicture}
\node[inner sep=0pt,circle,draw,thick,fill=gray] (0) at (0,0) {\small $10$};
\node[inner sep=1pt,circle,draw,thick] (1) at (-0.25,0.5) {\small $2$};
\node[inner sep=1pt,circle,draw,thick] (2) at (0.25,0.5) {\small $5$};
\draw[thick] (2) -- (0) -- (1);
\end{tikzpicture}\quad\begin{tikzpicture}
\node[inner sep=0pt,circle,draw,thick,fill=gray] (0) at (0,0) {\small $11$};
\node[inner sep=0pt,circle,draw,thick] (1) at (0,0.5) {\small $11$};
\draw[thick] (0) -- (1);
\end{tikzpicture}\quad\begin{tikzpicture}
\node[inner sep=0pt,circle,draw,thick,fill=gray] (0) at (0,0) {\small $12$};
\node[inner sep=1pt,circle,draw,thick] (1) at (-0.25,0.5) {\small $2$};
\node[inner sep=1pt,circle,draw,thick] (11) at (-0.25,1) {\small $2$};
\node[inner sep=1pt,circle,draw,thick] (2) at (0.25,0.5) {\small $3$};
\draw[thick] (2) -- (0) -- (1) -- (11);
\end{tikzpicture}\quad\begin{tikzpicture}
\node[inner sep=0pt,circle,draw,thick,fill=gray] (0) at (0,0) {\small $13$};
\node[inner sep=0pt,circle,draw,thick] (1) at (0,0.5) {\small $13$};
\draw[thick] (0) -- (1);
\end{tikzpicture}\quad\begin{tikzpicture}
\node[inner sep=0pt,circle,draw,thick,fill=gray] (0) at (0,0) {\small $14$};
\node[inner sep=1pt,circle,draw,thick] (1) at (-0.25,0.5) {\small $2$};
\node[inner sep=1pt,circle,draw,thick] (2) at (0.25,0.5) {\small $7$};
\draw[thick] (2) -- (0) -- (1);
\end{tikzpicture}\]

Removal of labels from the nodes of these trees forgets the numerical identities of the primes in the factorizations; this can be seen in the sequence of trees below corresponding to the one immediately above and noting that, for example, representations of primes 2, 3, 5, 7, 11 and 13 are identical, as are those of square-free semiprimes 6, 10 and 14:  

\[\begin{tikzpicture}
\node[inner sep=2pt,circle,draw,thick,fill=gray,label=-90:{\small $1$}] (0) at (0,0) {};
\end{tikzpicture}\quad\begin{tikzpicture}
\node[inner sep=2pt,circle,draw,thick,fill=gray,label=-90:{\small $2$}] (0) at (0,0) {};
\node[inner sep=2pt,circle,draw,thick] (1) at (0,0.5) {};
\draw[thick] (0) -- (1);
\end{tikzpicture}\quad\begin{tikzpicture}
\node[inner sep=2pt,circle,draw,thick,fill=gray,label=-90:{\small $3$}] (0) at (0,0) {};
\node[inner sep=2pt,circle,draw,thick] (1) at (0,0.5) {};
\draw[thick] (0) -- (1);
\end{tikzpicture}\quad\begin{tikzpicture}
\node[inner sep=2pt,circle,draw,thick,fill=gray,label=-90:{\small $4$}] (0) at (0,0) {};
\node[inner sep=2pt,circle,draw,thick] (1) at (0,0.5) {};
\node[inner sep=2pt,circle,draw,thick] (2) at (0,1) {};
\draw[thick] (0) -- (1) -- (2);
\end{tikzpicture}\quad\begin{tikzpicture}
\node[inner sep=2pt,circle,draw,thick,fill=gray,label=-90:{\small $5$}] (0) at (0,0) {};
\node[inner sep=2pt,circle,draw,thick] (1) at (0,0.5) {};
\draw[thick] (0) -- (1);
\end{tikzpicture}\quad\begin{tikzpicture}
\node[inner sep=2pt,circle,draw,thick,fill=gray,label=-90:{\small $6$}] (0) at (0,0) {};
\node[inner sep=2pt,circle,draw,thick] (1) at (-0.25,0.5) {};
\node[inner sep=2pt,circle,draw,thick] (2) at (0.25,0.5) {};
\draw[thick] (2) -- (0) -- (1);
\end{tikzpicture}\quad\begin{tikzpicture}
\node[inner sep=2pt,circle,draw,thick,fill=gray,label=-90:{\small $7$}] (0) at (0,0) {};
\node[inner sep=2pt,circle,draw,thick] (1) at (0,0.5) {};
\draw[thick] (0) -- (1);
\end{tikzpicture}\quad\begin{tikzpicture}
\node[inner sep=2pt,circle,draw,thick,fill=gray,label=-90:{\small $8$}] (0) at (0,0) {};
\node[inner sep=2pt,circle,draw,thick] (1) at (0,0.5) {};
\node[inner sep=2pt,circle,draw,thick] (2) at (0,1) {};
\draw[thick] (0) -- (1) -- (2);
\end{tikzpicture}\quad\begin{tikzpicture}
\node[inner sep=2pt,circle,draw,thick,fill=gray,label=-90:{\small $9$}] (0) at (0,0) {};
\node[inner sep=2pt,circle,draw,thick] (1) at (0,0.5) {};
\node[inner sep=2pt,circle,draw,thick] (2) at (0,1) {};
\draw[thick] (0) -- (1) -- (2);
\end{tikzpicture}\quad\begin{tikzpicture}
\node[inner sep=2pt,circle,draw,thick,fill=gray,label=-90:{\small $10$}] (0) at (0,0) {};
\node[inner sep=2pt,circle,draw,thick] (1) at (-0.25,0.5) {};
\node[inner sep=2pt,circle,draw,thick] (2) at (0.25,0.5) {};
\draw[thick] (2) -- (0) -- (1);
\end{tikzpicture}\quad\begin{tikzpicture}
\node[inner sep=2pt,circle,draw,thick,fill=gray,label=-90:{\small $11$}] (0) at (0,0) {};
\node[inner sep=2pt,circle,draw,thick] (1) at (0,0.5) {};
\draw[thick] (0) -- (1);
\end{tikzpicture}\quad\begin{tikzpicture}
\node[inner sep=2pt,circle,draw,thick,fill=gray,label=-90:{\small $12$}] (0) at (0,0) {};
\node[inner sep=2pt,circle,draw,thick] (1) at (-0.25,0.5) {};
\node[inner sep=2pt,circle,draw,thick] (11) at (-0.25,1) {};
\node[inner sep=2pt,circle,draw,thick] (2) at (0.25,0.5) {};
\draw[thick] (2) -- (0) -- (1) -- (11);
\end{tikzpicture}\quad\begin{tikzpicture}
\node[inner sep=2pt,circle,draw,thick,fill=gray,label=-90:{\small $13$}] (0) at (0,0) {};
\node[inner sep=2pt,circle,draw,thick] (1) at (0,0.5) {};
\draw[thick] (0) -- (1);
\end{tikzpicture}\quad\begin{tikzpicture}
\node[inner sep=2pt,circle,draw,thick,fill=gray,label=-90:{\small $14$}] (0) at (0,0) {};
\node[inner sep=2pt,circle,draw,thick] (1) at (-0.25,0.5) {};
\node[inner sep=2pt,circle,draw,thick] (2) at (0.25,0.5) {};
\draw[thick] (2) -- (0) -- (1);
\end{tikzpicture}\quad\begin{tikzpicture}
\end{tikzpicture}\]

Euler tours of the undecorated trees are then used to generate equivalent Dyck words corresponding to the arithmetic sequence $(1, 2, 3, \ldots)$; the sequence of the first fourteen such Dyck words is therefore
\begin{equation*}
(\epsilon, 10, 10, 1100, 10, 1010, 10, 1100, 1100, 1010, 10, 110010, 10, 1010),
\end{equation*}
with the absence of positional information of the primes in the undecorated trees being seen in their corresponding Dyck-word encodings.

These papers are clearly relevant to questions concerning learnability in artificial intelligence; \cite{contucci2025statistical} primarily characterizes the trees and corresponding Dyck words from a statistical perspective and speaks of a ``a corpus of planar rooted trees equivalently represented as Dyck words,'' while \cite{breccia2025testing} describes an experiment where ``a transformer network [$\ldots$] is trained from [the corpus] to subsequently test its predictive ability under next-word and masked-word prediction tasks.''

The results in  \cite{breccia2025testing} document partial learning of the grammar underlying the corpus, ``capturing non-trivial regularities and correlations,'' suggesting that ``learnability may extend beyond empircal data to the very structure of arithmetic.''  This leads me to wonder what results would come from similar experiments using corpora from lossless trees arising from the definition of the set of superrational numbers $\mathbb{S}$ as developed in Section~\ref{loc:AlgebraicDefOfS}.

I provide here two decorated trees side-by-side, the first from \cite{breccia2025testing} and the second arising from the superrational interpretation.  (For clarity, both depictions omit nodes with labels of 1 as children of nodes 2 and 5.) Their comparison suggests the idea underlying trees as implied by the sets defined in Section~\ref{loc:AlgebraicDefOfS} and realized in Definition~\ref{def:Nonsurjective_Superrational_Spelling} on page~\pageref{def:Nonsurjective_Superrational_Spelling}; it also illustrates why trees derived from the superrational interpretation are not lossy.

\begin{center}
\begin{tikzpicture}[thick]

\begin{scope}[xshift=-3cm]

  \node[circle,draw,fill=gray!63,inner sep=1pt] (L0) at (0,0) {\small $10$};

  \node[circle,draw,inner sep=1pt] (L1) at (-0.6,0.8) {\small $2$};
  \node[circle,draw,inner sep=1pt] (L2) at ( 0.6,0.8) {\small $5$};

  \draw (L1) -- (L0) -- (L2);
\end{scope}

\begin{scope}[xshift=3cm]
  \node[circle,draw,fill=gray!63,inner sep=1pt] (R0) at (0,0) {\small $10$};

  \node[circle,draw,inner sep=1pt] (R1) at (-1.0,0.8) {\small $2$};
  \node[circle,draw,inner sep=1pt] (R2) at ( 0.0,0.8) {\small $1$};
  \node[circle,draw,inner sep=1pt] (R3) at ( 1.0,0.8) {\small $5$};

  \draw (R1) -- (R0) -- (R2);
  \draw (R0) -- (R3);

\end{scope}

\end{tikzpicture}
\end{center}

From the two trees, which both represent the semiprime 10, we see the obvious distinction between them is that the tree on the left has no branch corresponding to $3^0 = 1$, because 3 is not a factor in the prime factorization of 10.   In other words, recursive factorization trees exemplified by the one shown on the left start from the accepted definition of prime factorization, extending that definition to yield a recursive system, whereas recursive factorization trees exemplified by the one on the right start with an \emph{extension} of the definition of prime factorization, with a further extension to yield recursivity.  Trees  from the superrational interpretation  incorporate sufficient positional information to identify which primes are being represented, at the cost of departing from the accepted definition of prime factorization.

Euler tours of trees derived from the superrational interpretation thus result in unique Dyck-word encodings of the nonnegative integers; the first nine are
\begin{equation*}
\epsilon, 10, 1100, 101100, 111000, 10101100, 11001100, 1010101100, 11011000.
\end{equation*}
\begin{remark}
Observe that 0 now has a representation, namely the empty string $\epsilon$.  This is a consequence of defining the superrational numbers as the closure of $\{0\}$ under prime-base exponentiation, multiplication and negation.
\end{remark}

I do not know what the effect of using such a Dyck word sequence to provide a corpus would be upon learnability.

On the one hand, the corpus described in \cite{contucci2025statistical,breccia2025testing} implies an equivalence relation, such that there is a partition where the primes are in one equivalence class, square-free semiprimes are in another, and sphenic numbers are in yet another, with these three equivalence classes only being examples.  Perhaps a corpus using Dyck words arising from the superrational representation would result in  \emph{reduced} learnability relative to the existing proposed benchmark due to obfuscation of the partition.  Also, the use of 10 to invest the representation with positional information comes at the cost of greatly increasing word size; the corpus would have to contain many more 1s and 0s to represent the same number of words as a lossy one.

On the other hand, I cannot rule out the possibility that even a corpus containing fewer Dyck words derived from the superrational interpretation would increase learnability due to its capture of any or all of the following under the same system of representation:
\begin{itemize}
  \item numerical identities of primes in the representation, as well as
  \item exact and unique representations of not only positve integers, but also
     \begin{itemize}
       \item zero (although zero, being represented by $\epsilon$, could not explicitly appear in the corpus),
       \item rational numbers,
       \item some algebraic irrationals, such as $\sqrt{2}$, and even
       \item some transcendental numbers, such as $2^{\sqrt{2}}$.
     \end{itemize}
\end{itemize}
 
Perhaps learnability could be further increased by training an LLM on both corpora.
\begin{remark}
Together with lossy factorization trees, the superrational interpretation implies a hierarchy of Dyck-word corpora with precisely defined arithmetic rules, with the information in the corpora increasing incrementally up the hierarchy. Seen in that light, I wonder whether such corpora might be used to supply a controlled environment for introducing incremental changes in a model’s internal statistics and comparing the results through increasing levels of arithmetic structure.
\end{remark}

\subsection{Can We Use Grammar-based Compression to Detect and Identify Patterns in RPF Word Sequences?}
Grammar-based compression algorithms such as Re-Pair \cite{bille2017practical} produce a context-free grammar for the string being compressed.  Input with lower information entropy (i.e., input that is less random) will have a higher compression ratio than input with higher information entropy. Suppose we  have two samples of input, these being of equal size but with one consisting of a concatenation of randomly-chosen Dyck superrationals, the other consisting of a concatenation of Dyck superrationals we conjecture to be consecutive terms in a sequence manifesting some pattern.  If the compression ratio obtained for the second input string is greater than that obtained for the first, we might take this as corroborating our conjecture.  If we find the difference  in compression ratios becomes more pronounced as the size of the input increases, we might take that to be an even stronger indication our conjecture is correct.

\subsection{Characterize Recursive Prime Factorizations Recognizable by Finite Automata}
Although $\mathcal{D}_{r_{\text{min}}}$ and $\mathcal{D}_{r_{\text{qmin}}}$ are not regular, many interesting subsets of them are.  For example, the set of square-free semiprimes (see below) is of particular importance in cryptography, being the modulus in the RSA public-key cryptosystem \cite[Chapter~8]{menezes2018handbook}.

We can employ regular expressions to describe such subsets.  In order to avoid confusion between grouping parentheses and symbols in the alphabet, at times we will encode nonzero RPF words in binary, using 1 for ${'('}$ and 0 for ${')'}$.  Under this scheme, of course, we will not be able to represent the empty word $\upepsilon$, for which there is no binary encoding, although we could overcome that by employing base-3 instead of base-2.

In all cases shown here, the correspondence of regular expressions to numerical sequences can readily be demonstrated by applying one or more of seven factual statements, each established earlier in the paper:

\begin{itemize}
\item Every $n \in \mathbb{N}$ is also a member of $\mathbb{S}$, with its minimal natural representation being identical to its minimal superrational representation.  For example, 8 is minimally represented by $(()(()))$ in both the natural and superrational numeral systems.
\item The $k$th chunk in a Dyck superrational word corresponds to the exponential factor with prime base $p_k$ for $k \in \{ 1,2,3, \ldots \}$.  Furthermore, the exponential factor is equal to $p_k^m$, where $m$ is the evaluation of the content of the chunk.  Thus, for example, the chunk $(())$ in  $()()(())$ corresponds to $p_3^{j}$, where ${j}$ is the number represented by ${()}$.
  \item 0 is represented by $\epsilon$; therefore, $() = (\epsilon) = 1$.
  \item Since any nonzero real number raised to the power 0 equals 1,  any noninflationary chunk equal to $()$ represents 1.
  \item Since any number raised to the power 1 equals that number, any chunk equal to $(())$ represents the prime number $p_k$, where $k$ is the position of the chunk in the word.  For example, $()()()(())$ represents $p_4 = 7$. 
  \item Since $()()$ represents $-1$, any chunk equal to $(()())$ represents the number $p_k^{-1}$, where $k$ is the position of the chunk in the word.  For example, $()()(()())$ represents $p_3^{-1} = 5^{-1} = 0.2$.
  \item Each of the six facts listed above recursively applies not only to Dyck words, but to their contents.
\end{itemize}

Here are some observations that follow.  As mentioned in the itemized list above, those which refer to sets of natural numbers could just as well be stated with reference to sets of superrational numbers, since natural numbers have identical representions in both minimal interpretations.

\begin{itemize}
\item  (10)*1100 recognizes the set of Dyck naturals representing prime numbers.  If words are placed in ascending order by semilength, the result encodes the sequence of prime numbers (OEIS A000040)  \cite{oeisA000040}.
\item  (10)*1100(10)*1100 recognizes the set of Dyck naturals representing square-free semiprime numbers.  If words are placed in ascending order by their arithmetic evaluations, the result encodes the sequence of square-free semiprimes  (OEIS A006881) \cite{oeisA006881}.
\item 10\textbar(1100)+ recognizes the set of Dyck naturals representing the set of primorial numbers $\{1, p_1,p_1p_2,p_1p_2p_3, \ldots\}$.  If words are placed in ascending order by semilength, the result encodes the sequence of primorials (OEIS A002110) \cite{oeisA002110}. 
\item (10)*11101000 recognizes the set of Dyck superrationals respresenting square roots of prime numbers.  If words are placed in ascending order by semilength, the result encodes the sequence of square roots of primes $(\sqrt{2}, \sqrt{3}, \sqrt{5}, \sqrt{7}, \ldots)$.
\end{itemize}

Let $k$ be a positive integer.  What is the characterization of the set $S_{k} \subset \mathbb{S}$ such that each member $s$ of $S_{k}$  is represented according to $\upalpha_{\mathbb{S}_{r_{\text{min}}}}$ by a word $w$ in some regular subset $D \subset \mathcal{D}_{r_{\text{qmin}}}$,  where no spelling $\upgamma_{\mathbb{S}_{r}}(s)$  has a recursion chain of length greater than $k$?

\section*{Acknowledgements}

The author thanks Dr.~Brian M.~Scott for his assistance in identifying the state transition table for the DFA used in Theorem~\ref{thm:RPFChomskySch}, as well as Kim Childress for carefully reading the Introduction and pointing out several errors in an earlier version of this paper. The author is also grateful to his wife, Runjuan, for her patience and support during the preparation of this work.
\pagebreak
\bibliographystyle{acm}
\bibliography{MathPaperReferences}
\end{document}